%% file: main.tex
\documentclass[article,11pt]{article}
\usepackage[utf8]{inputenc}
\usepackage{float,graphicx,verbatim,fullpage,hyperref,amssymb,amsmath,amsthm,enumerate,multicol, xspace,xcolor,mathtools,thmtools,thm-restate,cleveref,xspace,tikz,caption,subcaption,wasysym, enumitem}
\usepackage{subcaption}
\usetikzlibrary{shapes}
\tikzstyle{vertex}=[circle, draw, inner sep=1pt, minimum size=18pt]
\usepackage[margin=1in]{geometry}
\usepackage{algorithm, cite}
\usepackage[noend]{algpseudocode}
\usepackage{enumitem}
\usepackage{comment}

\usetikzlibrary{arrows.meta}
\tikzset{>={Latex[width=2mm,length=2mm]}}

\definecolor{darkgreen}{RGB}{50,150,50}
\def\final{0}  
\def\iflong{\iffalse}
\ifnum\final=0  
\newcommand{\knote}[1]{{\color{red}[{\tiny Karthik: \bf #1}]\marginpar{\color{red}*}}}
\newcommand{\enote}[1]{{\color{darkgreen}[{\small Elena: \bf #1}]\marginpar{\color{red}*}}}
\newcommand{\gnote}[1]{{\color{red}[{Gaby: \bf #1}]\marginpar{\color{red}*}}}
\newcommand{\snote}[1]{{\color{red}[{Shubhang: \bf #1}]\marginpar{\color{red}*}}}
\newcommand{\ynote}[1]{{\color{blue}[{\small Young-San: \bf #1}]\marginpar{\color{red}*}}}
\newcommand{\mnote}[1]{{\color{purple}[{Minshen: \bf #1}]\marginpar{\color{red}*}}}
\newcommand{\todo}[1]{{\color{red}[{ TODO: \bf #1}]\marginpar{\color{red}*}}}
\else 
\newcommand{\knote}[1]{}
\newcommand{\enote}[1]{}
\newcommand{\gnote}[1]{}
\newcommand{\snote}[1]{}
\newcommand{\ynote}[1]{}
\newcommand{\mnote}[1]{}
\newcommand{\todo}[1]{}
\fi  

\def\Z{\mathbb{Z}}

\newtheorem{theorem}{Theorem}[section]

\newtheorem{lemma}[theorem]{Lemma}
\newtheorem{claim}[theorem]{Claim}
\newtheorem{corollary}[theorem]{Corollary}

\theoremstyle{definition}
\newtheorem{definition}[theorem]{Definition}

\newcommand{\umaxBT}{\textsc{Undir\-Max\-Binary\-Tree}\xspace}

\newcommand{\dmaxBT}{\textsc{Dir\-Max\-Binary\-Tree}\xspace}
\newcommand{\rdmaxBT}{\textsc{rooted-\-Dir\-Max\-Binary\-Tree}\xspace}
\newcommand{\dagmaxBT}{\textsc{DAG\-Max\-Binary\-Tree}\xspace}

\newcommand{\kBT}{$k$-\textsc{Binary\-Tree}\xspace}
\newcommand{\kpath}{$k$-\textsc{Path}\xspace}

\input{newcommands}

\begin{document}

\title{The Maximum Binary Tree Problem}

\author{Karthekeyan Chandrasekaran\thanks{University of Illinois, Urbana-Champaign, Email: karthe@illinois.edu}
\and Elena Grigorescu\thanks{Purdue University, Email: \{elena-g, kulkar17, lin532, zhu628\}@purdue.edu}
\and Gabriel Istrate\thanks{West University of Timi\c{s}oara, Romania, and the e-Austria Research Institute. Email: gabrielistrate@acm.org}\\
\and Shubhang Kulkarni\footnotemark[2]
\and Young-San Lin\footnotemark[2]
\and Minshen Zhu\footnotemark[2]
}
\date{\today}

\maketitle

\input{abstract}

\input{intro}


\input{preliminaries}

\input{arithmetic-circuit.tex}

\input{dags-hardness}

\input{dag-lp-integrality-gap.tex}

\input{undirected-hardness}

\input{bipartite-permutation-graphs.tex}

\input{conclusions}

\input{acknowledgement}

\bibliographystyle{amsplain}
\bibliography{reference}

\end{document}

%% file: newcommands.tex
\usepackage{environ}

\newcommand{\abs}[1]{\left\vert#1\right\vert}
\newcommand{\set}[1]{\left\{#1\right\}}
\newcommand{\tuple}[1]{\left(#1\right)} 
\newcommand{\eps}{\varepsilon}
\newcommand{\tp}{\tuple}
 \renewcommand{\P}{\mathbf{P}}
\newcommand{\NP}{\mathbf{NP}}
\newcommand{\DTIME}[1]{\mathbf{DTIME}\tp{#1}}

\def\*#1{\mathbf{#1}}
\def\+#1{\mathcal{#1}}

\newcommand{\ceilfit}[1]{\left\lceil{#1}\right\rceil}

\newcommand{\itn}[1]{^{(#1)}}

\NewEnviron{problem}[1]{%
	\begin{center}\fbox{\parbox{6in}{%
				{\centering\scshape #1\par}%
				\parskip=1ex
				\everypar{\hangindent=1em}%
				\BODY
}}\end{center}}

\newcommand{\mbt}{\textsf{MBT}}
\newcommand{\opt}{\textsf{OPT}}
\newcommand{\GL}[1]{\mathbb{F}_{#1}}
\newcommand{\poly}{\ensuremath{\mathsf{poly}}}

%% file: abstract.tex
\begin{abstract}

We introduce and investigate the approximability of the {\em maximum binary tree problem} (MBT) in directed and undirected graphs. The goal in MBT is to find a maximum-sized binary tree in a given graph. MBT is a natural variant of the well-studied longest path problem, since both can be viewed as finding a maximum-sized tree of bounded degree in a given graph. 

The connection to longest path motivates the study of MBT in directed acyclic graphs (DAGs), since the longest path problem is solvable efficiently in DAGs. In contrast, we show that MBT in DAGs is in fact hard: it has no efficient $\exp(-O(\log n/ \log \log n))$-approximation algorithm under the exponential time hypothesis, where $n$ is the number of vertices in the input graph. In undirected graphs, we show that MBT has no efficient  $\exp(-O(\log^{0.63}{n}))$-approximation under the exponential time hypothesis. Our inapproximability results rely on self-improving reductions and structural properties of binary trees. We also show constant-factor inapproximability assuming $\P\neq \NP$.

In addition to inapproximability results, we present algorithmic results along two different flavors: (1) We design a randomized algorithm to verify if a given directed graph on $n$ vertices contains a binary tree of size $k$ in $2^k \poly(n)$ time. (2) Motivated by the longest heapable subsequence problem, introduced by Byers, Heeringa, Mitzenmacher, and Zervas, \emph{ANALCO 2011}, which is equivalent to MBT in \emph{permutation DAGs}, we design efficient algorithms for MBT in bipartite permutation graphs.

\end{abstract}

%% file: intro.tex
\section{Introduction}
A general degree-constrained subgraph problem asks for an optimal subgraph of a given graph with specified properties while also satisfying degree constraints on all vertices. 
Degree-constrained subgraph problems have numerous applications in the field of network design and consequently, have been studied extensively in the algorithms and approximation literature \cite{Gab83,EFRS90,FR94,RMRRH01,ADR08, LNSS09, SL15}. In this work, we introduce and study the \emph{maximum binary tree problem} in directed and undirected graphs. In the maximum binary tree problem (MBT), we are given an input graph $G$ and the goal is to find a \emph{binary tree} in $G$ with maximum number of vertices.

Our first motivation for studying MBT arises from the viewpoint that it is 
a variant of the longest path problem: In the longest path problem, the goal is to find a maximum-sized tree in which every vertex has degree at most $2$. In MBT, the goal is to find a maximum-sized tree in which every vertex has degree at most $3$. Certainly, one may generalize both these problems to finding a {\em maximum-sized degree-constrained tree} in a given graph. In this work we focus on binary trees; however, all our results extend to the maximum-sized degree-constrained tree problem for \emph{constant} degree bound.

Our second motivation for studying MBT is its connection to the \emph{longest heapable subsequence problem} introduced by Byers, Heeringa, Mitzenmacher, and Zervas \cite{byers2010heapable}. Let $\sigma = (\sigma_1 , \sigma_2, \ldots, \sigma_n)$ be a permutation  on $n$ elements. Byers et al. define a subsequence (not necessarily contiguous) of $\sigma$ to be \emph{heapable} if the elements of the subsequence can be sequentially inserted to form a binary \emph{min-heap} data structure. Namely, insertions subsequent to the first element, which takes the root position, happen below previously placed elements. The longest heapable subsequence problem asks for a maximum-length heapable subsequence of a given sequence. This generalizes the well-known longest increasing subsequence problem. Porfilio \cite{porfilio2015combinatorial} showed that the longest heapable subsequence problem is equivalent to MBT in permutation directed acyclic graphs (abbreviated \emph{permutation DAGs}): a permutation DAG associated with the sequence $\sigma$ is obtained by introducing a vertex $u_i$ for every sequence element $\sigma_i$, and arcs $(u_i,u_j)$ for every pair $(i,j)$ such that $i>j$ and $\sigma_i \ge \sigma_j$. On the other hand, for sequences of intervals the maximum binary problem is easily solvable by a greedy algorithm \cite{balogh2017heapability} (see also \cite{istrate2016heapability} for further results and open problems on the heapability of partial orders). These results motivate the study of MBT in restricted graph families.

We now formally define MBT in undirected graphs, which we denote as \umaxBT. A \emph{binary tree} of an {\em undirected} graph $G$ is a subgraph $T$ of $G$ that is connected and acyclic with the degree of $u$ in $T$ being at most $3$ for every vertex $u$ in $T$.  In \umaxBT, the input is an undirected graph $G$ and the goal is to find a binary tree in $G$ with maximum number of vertices. In the rooted variant of this problem, the input is an undirected graph $G$ along with a specified root vertex $r$ and the goal is to find a binary tree containing $r$ in $G$ with maximum number of vertices such that the degree of $r$ in the tree is at most $2$. We focus on the unrooted variant of the problem and mention that it reduces to the rooted variant.
We emphasize that a binary tree $T$ of $G$ is not necessarily spanning (i.e., may not contain all vertices of the given graph).
The problem of verifying whether a given undirected graph has a {\em spanning} binary tree is $\NP$-complete. This follows by a reduction from the Hamiltonian path problem: Given an undirected graph $G=(V,E)$, create a pendant vertex $v'$ adjacent to $v$ for every vertex $v\in V$. The resulting graph has a spanning binary tree if and only if $G$ has a Hamiltonian path.

Next, we formally define MBT in directed graphs. A \emph{tree of a directed graph} $G$ is a subgraph $T$ of $G$ such that $T$ is acyclic and has a unique vertex, termed as the root, with the property that every vertex $v$ in $T$ has a unique directed path \emph{to the root} in $T$. A \emph{binary tree of a directed graph} $G$ is a tree $T$ such that the incoming-degree of every vertex $u$ in $T$ is at most $2$ while the outgoing-degree of every vertex $u$ in $T$ is at most $1$. In the rooted variant of the maximum binary tree problem for directed graphs, the input is a directed graph $G$ along with a specified root $r$ and the goal is to find an $r$-rooted binary tree $T$ in $G$ with maximum number of vertices. The problem of verifying whether a given directed graph has a \emph{spanning} binary tree is $\NP$-complete (by a similar reduction as that for undirected graphs). 

The connection to the longest path problem as well as the longest heapable subsequence problem motivates the study of the maximum binary tree problem in directed acyclic graphs (DAGs). In contrast to directed graphs, the longest path problem in DAGs can be solved in polynomial-time (e.g., using dynamic programming or LP-based techniques). Moreover, verifying whether a given DAG contains a spanning binary tree is solvable in polynomial-time using the following characterization: a given DAG on vertex set $V$ contains a spanning binary tree if and only if the partition matroid corresponding to the in-degree of every vertex being at most two and the partition matroid corresponding to the out-degree of every vertex being at most one have a common independent set of size $|V|-1$. These observations raise the intriguing possibility of solving the maximum binary tree problem in DAGs in polynomial-time. For this reason, we focus on DAGs within the family of directed graphs in this work. We denote the maximum binary tree problem in DAGs as \dagmaxBT.

The rooted and the unrooted variants of the maximum binary tree problem in DAGs are polynomial-time equivalent by simple transformations. Indeed, the unrooted variant can be solved by solving the rooted variant for every choice of the root. To see the other direction, suppose we would like to find a maximum $r$-rooted binary tree in a given DAG $G=(V,E)$. 
Then, we discard from $G$ all outgoing arcs from $r$ and all vertices that cannot reach $r$ (i.e., we consider the sub-DAG induced by the descendents of $r$) and find an unrooted maximum binary tree in the resulting DAG. If this binary tree is rooted at a vertex $r' \neq r$, then it can be extended to an $r$-rooted binary tree by including an arbitrary $r' \rightarrow r$ path $P$---since the graph is a DAG, any such path $P$ will not visit a vertex that is already in the tree (apart from $r'$). The equivalence is also approximation preserving. For this reason, we only study the rooted variant of the problem in DAGs.

We present inapproximability results for MBT in DAGs and undirected graphs. On the algorithmic side, we show that MBT in directed graphs is fixed-parameter tractable when parameterized by the solution size. 
We observe that the equivalence of the longest heapable subsequence to MBT in permutation DAGs motivates the study of MBT even in restricted graph families. As a first step towards understanding MBT in permutation DAGs, we 
design an algorithm for 
bipartite permutation graphs. 
We use a variety of tools including self-improving and gadget reductions for our inapproximability results, and algebraic and structural techniques for our algorithmic results. \\


\subsection{Related work} \label{appendix-related-work}
Degree-constrained subgraph problems appeared as early as 1978 in the textbook of Garey and Johnson \cite{GJ-book} and have garnered plenty of attention in the approximation community  \cite{Gab83,EFRS90,FR94,RMRRH01,ADR08, KKN08, LNSS09, SL15}. A rich line of works have addressed the minimum degree spanning tree problem as well as the minimum cost degree-constrained spanning tree problem leading to powerful rounding techniques and a deep understanding of the spanning tree polytope \cite{CRRT09a,CRRT09b,KR02, KR05,Goe06,SL15,FR94}. Approximation and bicriteria approximation algorithms for the counterparts of these problems in directed graphs, namely degree-constrained arborescence and min-cost degree-constrained arborescence, have also been studied in the literature \cite{BKN09}. 

In the \emph{maximum-edge} degree-constrained connected subgraph problem, the goal is to find a connected degree-constrained subgraph of a given graph with maximum number of edges. This problem does not admit a PTAS \cite{APPSS09} and has been studied from the perspective of fixed-parameter tractability \cite{ASS08}. 
MBT could be viewed as a \emph{maximum-vertex} degree-constrained connected subgraph problem, where the goal is to maximize the number of \emph{vertices} as opposed to the number of \emph{edges}---the degree-constrained connected subgraph maximizing the number of vertices may be assumed to be acyclic and hence, a tree. 
It is believed that the \emph{connectivity constraint} makes the maximum-edge degree-constrained connected subgraph problem to become extremely difficult to approximate. Our results formalize this belief when the objective is to maximize the number of vertices. 

Switching the objective with the constraint in the maximum-vertex degree-constrained connected subgraph problem leads to the minimum-degree $k$-tree problem: here the goal is to find a minimum degree subgraph that is a tree with at least $k$ vertices. Minimum degree $k$-tree admits a $O(\sqrt{(k/\Delta^*)\log{k}})$-approximation, where $\Delta^*$ is the optimal degree and does not admit a $o(\log{n})$-approximation \cite{KKN08}. We note that the hardness reduction here (from set cover) crucially requires the optimal solution value $\Delta^*$ to grow with the number $n$ of vertices in the input instance, and hence, does not imply any hardness result for input instances in which $\Delta^*$ is a constant. Moreover, the approximation result implies that a tree of degree $O(\sqrt{k\log{k}})$ containing $k$ vertices can be found in polynomial time if the input graph contains a constant-degree tree with $k$ vertices.

We consider the maximum binary tree problem to be a generalization of the longest path problem as both can be viewed as asking for a maximum-sized degree-constrained connected acyclic subgraph. The longest path problem in undirected graphs admits an $\Omega\tp{{(\log{n}/\log\log{n})^2}/{n}}$-approximation \cite{BHK04}, but it is APX-hard and does not admit a $2^{-O\tp{\log^{1-\eps}{n}}}$-approximation for any constant $\eps>0$ unless $\NP \subseteq \DTIME{2^{\log^{O(1/\eps)}{n}}}$ \cite{karger1997approximating}. Our hardness results for the max binary tree problem in undirected graphs bolsters this connection. The longest path problem in directed graphs is much harder: For every $\eps >0$ it cannot be approximated to within a factor of $1/n^{1-\eps}$ unless $\P = \NP$, and it cannot be approximated to within a factor of $(\log^{2+\eps} n)/n$ under the Exponential Time Hypothesis \cite{BHK04}. However, the longest path problem in DAGs is solvable in polynomial time. Our hardness results for the max binary tree problem in DAGs are in stark contrast to the polynomial-time solvability of the longest path problem in DAGs. 

On the algorithmic side, the color-coding technique introduced by Alon, Yuster, and Zwick \cite{alon1995color} can be used to decide whether an undirected graph $G=(V,E)$ contains a copy of a bounded treewidth pattern graph $H=(V_H, E_H)$ where $|V_H|=O\tp{\log|V|}$, and if so, then find one in polynomial time. The idea here is to randomly color the vertices of $G$ by $O\tp{\log|V|}$ colors and to find a maximum colorful copy of $H$ using dynamic programming. We note that the same dynamic programming approach can be modified to find a maximum colorful binary tree. This algorithm can be derandomized, thus leading to a deterministic $\Omega\tp{({1}/{n})\log{n}}$-approximation to \umaxBT. 

In parameterized complexity, designing algorithms with running time $\beta^k \poly(n)$ ($\beta > 1$ is a constant) for problems like \kpath and $k$-\textsc{Tree} is a central topic.  For \kpath, the color-coding technique mentioned above already implies a $(2e)^k \poly(n)$-time algorithm. Koutis \cite{koutis2008faster} noticed that \kpath can be reduced to detecting whether a given polynomial contains a multilinear term. Using algebraic methods for the latter problem, Koutis obtained a $2^{1.5k} \poly(n)$ time algorithm for \kpath. This was later improved by Williams \cite{williams2009finding} to $2^k \poly(n)$. The current state-of-art algorithm is due to Bj\"orklund, Husfeldt, Kaski and Koivisto 
\cite{bjorklund2017narrow}, which is also an algebraic algorithm with running time $1.66^k \poly(n)$.
All of these algorithms are randomized. 
Our study of the \kBT problem, which is the problem of deciding whether a given graph $G$ contains a binary tree of size at least $k$, is inspired by this line of results. 

Several $\NP$-hard problems are known to be solvable in specific families of graphs. Bipartite permutation graphs 
is one such family which is 
known to exhibit this behaviour 
\cite{Efficient_Alg_Longest_Path, smith2011minimum,spinrad1987bipartite,KLOKS1998313}.
Our polynomial-time solvability result for these families of graphs crucially identifies the existence of structured optimal solutions to reduce the search space and solves the problem over this reduced search space.

\subsection{Our contributions}


\subsubsection{Inapproximability results}

\noindent \textbf{Directed graphs.}
We first focus on directed graphs and in particular, on directed acyclic graphs. 
It is well-known that the longest path problem in DAGs is solvable in polynomial-time. In contrast, we show that \dagmaxBT does not even admit a constant-factor approximation. Furthermore, if \dagmaxBT admitted a polynomial-time $\exp\tp{-O\tp{\log{n}/\log\log{n}}}$-approximation algorithm then the Exponential Time Hypothesis would be violated. 

\begin{restatable}{theorem}{dagmaxBTnoApprox}\label{theorem:dagmaxBT-no-const-approx}
We have the following inapproximability results for \dagmaxBT on $n$-vertex input graphs: 
\begin{enumerate}
\item \label{theorem:dagmaxBT-no-const-approx_1} \dagmaxBT does not admit a polynomial-time constant-factor approximation assuming $\P \neq \NP$.
\item If \dagmaxBT admits a polynomial-time $\exp\tp{-O\tp{\log{n}/\log\log{n}}}$-approximation, then $\NP \subseteq \DTIME{\exp\tp{O\tp{\sqrt{n}}}}$, refuting the Exponential Time Hypothesis.
\item For any $\eps>0$, if \dagmaxBT admits a quasi-polynomial time $\exp\tp{-O\tp{\log^{1-\eps}{n}}}$-approximation, then $\NP \subseteq \DTIME{\exp\tp{\log^{O\tp{1/\eps}}{n}}}$, thus refuting the Exponential Time Hypothesis.
\end{enumerate}

\end{restatable}

\paragraph{LP-based approach.} The longest path problem in DAGs can be solved using a  linear program (LP) based on cut constraints. Based on this connection, an integer program (IP) based on cut constraints can be formulated for \dagmaxBT. In Section \ref{sec:IP-gap-dags}, we show that the LP-relaxation of this cut-constraints-based-IP has an integrality gap of $\Omega(n^{1/3})$ in $n$-vertex DAGs.\\

\noindent \textbf{Undirected graphs.} 
Next, we turn to undirected graphs. We show that \umaxBT does not have a constant-factor approximation and does not admit a quasi-polynomial-time exp$(-O(\log^{0.63}{n}))$-approximation under the Exponential Time Hypothesis. 
\begin{restatable}{theorem}{umaxBTnoApprox} \label{theorem:umaxBT-no-const-approx}
We have the following inapproximability results for \umaxBT on $n$-vertex input graphs: 
\begin{enumerate}
    \item \label{theorem:umaxBT-no-const-approx_1} \umaxBT 
does not admit a polynomial-time constant-factor approximation assuming $\P \neq \NP$.
\item For $c=\log_3{2}$ and any $\eps > 0$, if \umaxBT admits a quasi-polynomial time $\exp\tp{-O\tp{\log^{c-\eps}{n}}}$-approximation, then $\NP \subseteq \DTIME{\exp\tp{\log^{O\tp{{1/\eps}}}{n}}}$, thus refuting the Exponential Time Hypothesis. 
\end{enumerate} 
\end{restatable}

We summarize our hardness results for MBT on various graph families in Table \ref{hardness} and contrast them with the corresponding known hardness results for the longest path problem on those families.

\begin{table}[H]
\begin{center}
\def\arraystretch{1.2}
\footnotesize
\begin{tabular}{|*4{l|}}
\hline
\textbf{Family} & \textbf{Assumption} & \textbf{Max Binary Tree} & \textbf{Longest Path}  \\
\hline
DAGs & $\P \neq \NP$ & No poly-time $\Omega(1)$-apx (Thm \ref{theorem:dagmaxBT-no-const-approx}) & {\color{gray}Poly-time solvable} \\
\cline{2-4}
 & ETH & No poly-time $\exp(-O(\frac{\log n}{\log\log n}))$-apx & {\color{gray}Poly-time solvable} \\
 &     & No quasi-poly-time & \\
 &     & \quad $\exp(-O(\log^{1-\eps} n))$-apx (Thm \ref{theorem:dagmaxBT-no-const-approx}) & \\
\hline
Directed & $\P\neq \NP$ & Same as DAGs & {\color{gray} No poly-time $\dfrac{1}{n^{1-\eps}}$-apx \cite{BHK04}} \\
\cline{2-4}
 & ETH & Same as DAGs & {\color{gray} Same as $\P\neq \NP$}\\
\hline
Undirected & $\P \neq \NP$ & No poly-time $\Omega(1)$-apx (Thm \ref{theorem:umaxBT-no-const-approx}) & {\color{gray} No poly-time $\Omega(1)$-apx \cite{karger1997approximating}} \\
\cline{2-4}
 & ETH & No quasi-poly-time & {\color{gray} No quasi-poly-time} \\
 &     & \quad $\exp(-O(\log^{0.63-\eps} n))$-apx (Thm \ref{theorem:umaxBT-no-const-approx}) & \quad {\color{gray} $\exp(-O(\log^{1-\eps} n))$-apx \cite{karger1997approximating}} \\
\hline
\end{tabular}
\caption{Summary of inapproximability results. Here, $n$ refers to the number of vertices in the input graph and $\epsilon$ is any positive constant. We include the known results for longest path for comparison. Text in gray refer to known results while text in black refer to our contributions.}
\label{hardness}
\end{center}
\end{table}

\subsubsection{Algorithmic results}
\textbf{Fixed-parameter tractability.} We denote the decision variant of MBT as \kBT---here the goal is to verify if a given graph contains a binary tree with at least $k$ vertices.  
Since \kBT is $\NP$-hard when $k$ is part of the input, it is desirable to have an algorithm that runs in time $f(k)\poly(n)$ (i.e., a fixed parameter algorithm parameterized by the solution size). Our first algorithmic result achieves precisely this goal. 
Our algorithm is based on algebraic techniques. 
\begin{restatable}{theorem}{k-binary-tree}\label{theorem:k-binary-tree}
There exists a randomized algorithm that takes a directed graph $G=(V,E)$, a positive integer $k$, and a real value $\delta \in (0,1)$ as input, runs in time $2^k \poly(|V|)\log(1/\delta)$ and
\begin{enumerate}
\item outputs 'no' if $G$ does not contain a binary tree of size $k$;
\item outputs a binary tree of size $k$ with probability $1-\delta$ if $G$ contains one.
\end{enumerate}
\end{restatable}

\noindent \textbf{Bipartite permutation graphs.} Next, motivated by its connection to the max heapable subsequence problem, we study MBT in \emph{bipartite permutation graphs}. A bipartite permutation graph is a permutation graph (undirected) which is also bipartite. We show that bipartite permutation graphs admit an efficient algorithm for MBT. Our algorithm exploits structural properties of bipartite permutation graphs.  
We believe that these structural properties could be helpful in solving MBT in permutation graphs which, in turn, could provide key insights towards solving MBT in permutation DAGs. 

\begin{restatable}{theorem}{bipermALG}\label{theorem:bipartite-permutation-graphs}
There exists an algorithm to solve \umaxBT in $n$-vertex bipartite permutation graphs that runs in time $O(n^3)$.
\end{restatable}

We summarize our algorithmic results for MBT in Table \ref{algorithmic} and contrast them with the corresponding best known bounds for the longest path problem.

\begin{table}[ht!]
\begin{center}
\small
\def\arraystretch{1.18}
\begin{tabular}{|*3{l|}}
\hline
\textbf{Problem} & \textbf{Max Binary Tree} & \textbf{Longest Path}  \\
\hline
FPT parameterized by solution size (Dir.)
& $2^k\poly(n)$-time (Thm \ref{theorem:k-binary-tree}) & {\color{gray} $1.66^k\poly(n)$-time \cite{bjorklund2017narrow}}\\
\hline
Bipartite permutation graphs (Undir.) & $O(n^3)$-time (Thm \ref{theorem:bipartite-permutation-graphs}) & {\color{gray} $O(n)$-time \cite{Efficient_Alg_Longest_Path}} \\
\hline
\end{tabular}
\caption{Summary of algorithmic results. Here, $n$ refers to the number of vertices in the input graph. We include the known results for longest path for comparison. Text in gray refer to known results while text in black refer to our contributions.}
\label{algorithmic}
\end{center}
\end{table}
We remark again that our inapproximability as well as algorithmic results are also applicable to the maximum degree-constrained tree problem for larger, but constant degree constraint. We focus on the degree constraint corresponding to binary trees for the sake of simplicity in exposition. 

\subsection{Proof techniques}

In this section, we outline the techniques underlying our results. 

\subsubsection{Inapproximability results} 
At a very high level, our inapproximability results for MBT rely on the proof strategy for hardness of longest path due to Karger, Motwani, and Ramkumar \cite{karger1997approximating}, which has two main steps: (1) a {\em self-improving} reduction whose amplification implies that a constant-factor approximation immediately leads to a PTAS,   and (2) a proof that there is no PTAS. 
However, we achieve both these steps in a completely different manner compared to the approach of Karger, Motwani, and Ramkumar. Both their steps are tailored for the longest path problem, but fail for the maximum degree-constrained tree problem. Our results for MBT require several novel ideas, as described next.



Karger, Motwani and Ramkumar's 
{self-improving} reduction for the longest path proceeds as follows: given an undirected
graph $G$, they obtain a squared graph $G^2$ by replacing each edge $\{u,v\}$ of $G$ with a copy of $G$ by adding edges from $u$ and $v$ to all vertices in that edge copy. Let $OPT(G)$ be the length of the longest path in $G$. They make the following two observations:  
    Obs (i) $OPT(G^2)\geq OPT(G)^2$ and 
    Obs (ii) a path in $G^2$ of length at least $\alpha OPT(G^2)$ can be used to recover a path in $G$ of length at least $\sqrt{\alpha} OPT(G)$.
The first observation is because we can extend any path $P$ in $G$ into a path of length $|E(P)|^2$ by traversing each edge copy also along $P$. The second observation is because for any path $P_2$ in $G^2$ either $P_2$ restricted to some edge copy of $G$ is a path of length at least $\sqrt{|E(P_2)|}$ or projecting $P_2$ to $G$ (i.e., replacing each sub-path of $P_2$ in each edge copy by a single edge) gives a path of length at least $\sqrt{|E(P_2)|}$. 
We note that a similar construction of the squared graph for directed graphs also has the above mentioned observations: replace each directed arc $(u,v)$ of $G$ with a copy of $G$ by adding arcs from $u$ to all vertices in that edge copy and from all vertices in that edge copy to $v$. 

In order to obtain inapproximability results for the maximum binary tree problem, 
we first introduce different constructions for the squared graph in the self-improving reduction compared to the ones by Karger et al. 
Moreover, our constructions of the squared graph differ substantially between undirected and directed graphs. Interestingly, our constructions also generalize naturally to the max degree-constrained tree problem. 
Secondly, although our reduction for showing the lack of PTAS in undirected graphs for MBT is also from TSP$(1, 2)$, it is completely different from that of Karger et al. and, once again, generalizes to the max degree-constrained tree problem. Thirdly, we show the lack of PTAS in DAGs for MBT by reducing from the max $3$-coloring problem. This reduction is altogether new---the reader might recall that the longest path problem in DAGs is solvable in polynomial-time, so there cannot be a counterpart of this step (i.e., lack of PTAS in DAGs) for longest path.
We next present further details underlying our proofs. \\

\noindent \textbf{Self-improving reduction for directed graphs.} 
We focus on the rooted variant of MBT in directed graphs. 
We first assume that the given graph $G$ contains a source (if not, adding such a source vertex with arcs to all the vertices changes the optimum only by one). In contrast to the squared graph described above (i.e., instead of adding edge copies), we replace every vertex in $G$ by a copy of $G$ (that we call as a vertex copy) and for every arc $(u,v)$ in $G$, we add an arc from the root node of the vertex copy corresponding to $u$ to the source node of the vertex copy corresponding to $v$. Finally, we declare the root node of the root vertex copy to be the root node of $G^2$. Let $\alpha \in (0,1]$ and $OPT(G)$ be the number of vertices in the maximum binary tree in $G$. With this construction of the squared graph, we show that (1) $OPT(G^2)\geq OPT(G)^2$ and (2) an $\alpha$-approximate rooted binary tree $T_2$ in $G^2$ can be used to recover a rooted binary tree $T_1$ in $G$ which is a $\sqrt{\alpha}$-approximation. We emphasize that if $G$ is a DAG, then the graph $G^2$ obtained by this construction is also a DAG.\\

\noindent \textbf{Inapproximability for DAGs.} 
In order to show the constant-factor inapproximability result for DAGs, it suffices to show that there is no PTAS (due to the self-improving reduction for directed graphs described above). We show the lack of a PTAS in DAGs by reducing from the max $3$-coloring problem in $3$-colorable graphs. It is known that this problem is APX-hard---in particular, there is no polynomial-time algorithm to find a coloring that colors at least $32/33$-fraction of the edges properly \cite{guruswami2013improved}. Our reduction encodes the coloring problem into a \dagmaxBT instance in a way that recovers a consistent coloring for the vertices while also being proper for a large fraction of the edges. 
Our ETH-based inapproximability result is also a consequence of this reduction in conjunction with the self-improving reduction. We again emphasize that there is no counterpart of APX-hardness in DAGs for max binary tree in the longest path literature.\\

\noindent \textbf{Self-improving reduction for undirected graphs.}
For \umaxBT, the self-improving reduction is more involved. Our above-mentioned reduction for \dmaxBT heavily exploits the directed nature of the graph (e.g., uses source vertices) and hence, is not applicable for undirected graphs. Moreover, the same choice of squared graph $G^2$ as Karger et al.  \cite{karger1997approximating} fails since 
Obs (ii) does not hold any more: the tree $T_2$ restricted to each edge copy may not be a tree (but it will be a forest). However, we observe that $T_2$ restricted to each edge copy may result in a forest with up to four binary trees in it. This observation and a more careful projection can be used to recover a tree of size at least  $\sqrt{|V(T_2)|}/4$ (let us call this weakened 
Obs (ii)). Yet, weakened 
Obs (ii)
is insufficient for a self-improving reduction. One approach to fix this would be to construct a \emph{different squared graph} $G^{\XBox 2}$ that strengthens 
Obs (i) 
to guarantee that $OPT(G^{\XBox 2})\geq 16OPT(G)^2$ while still allowing us to recover a binary tree of size $\sqrt{|V(T_2)|}/4$ in $G$ from a binary tree $T_2$ in $G^{\XBox 2}$. Such a strengthened 
Obs (i) 
coupled with weakened 
Obs (ii) 
would complete the self-improving reduction. Our reduction is a variant of this approach: we introduce a construction of the squared-graph that strengthens 
Obs (i) 
by a factor of $2$ while also weakening 
Obs (ii) 
only by a factor of $2$. We prove these two properties of the construction by relying on a handshake-like property of binary trees which is a relationship between the number of nodes of each degree and the total number of nodes in the binary tree.\\

\noindent \textbf{Inapproximability for undirected graphs.}
In order to show the constant-factor inapproximability result, it suffices to show that there is no PTAS (due to the self-improving reduction). We show the lack of a PTAS by reducing from TSP$(1,2)$. We mention that Karger, Motwani, and Ramkumar \cite{karger1997approximating} also show the lack of a PTAS for the longest path problem by reducing from TSP$(1,2)$. However, our reduction is much different from their reduction. Our reduction mainly relies on the fact that if we add a pendant node to each vertex of a graph $G$ and obtain a binary tree $T$ that has a large number of such pendants, then the binary tree restricted to $G$ cannot have too many nodes of degree three. 
Our ETH-based inapproximability result is also a consequence of this reduction in conjunction with the self-improving reduction.

\subsubsection{Algorithmic results.}

\noindent \textbf{A $2^k \poly(n)$ time algorithm for \kBT.}
The proof of this result is inspired by the algebrization technique introduced in \cite{koutis2008faster, williams2009finding, koutis2009limits} for designing randomized algorithms for 
\kpath and $k$-\textsc{Tree}---in \kpath, the goal is to recover a path of length $k$ in the given graph while $k$-\textsc{Tree} asks to recover a \emph{given tree} on $k$ vertices in the given graph. Their idea is to encode a path (or the given tree) as a multilinear monomial term 
in a carefully constructed polynomial, which is efficiently computable using an arithmetic circuit. 
Then, a result due to Williams \cite{williams2009finding} is used to verify if the constructed polynomial contains a multilinear term---Williams' result gives an efficient randomized algorithm, which on input  a small circuit that computes the polynomial,  outputs `yes' if a multilinear term exists in the sum of products representation of the input polynomial, and `no' otherwise. 
The subgraph that is sought may then be extracted using an additional pass over the graph.
Our main technical contribution is the construction of a polynomial $P_G$ whose multilinear terms correspond to binary trees of size $k$ in $G$ and which is efficiently computable by an arithmetic circuit. 
We remark that the polynomial constructions in previous results do not readily generalize for our problem. Our key contribution is the construction of a suitable polynomial, based on a carefully designed recursion.\\ 

\noindent \textbf{Efficient algorithm for bipartite permutation graphs.} 
Our main structural insight for bipartite permutation graphs is that there exists a maximum binary tree which is \emph{crossing-free} with respect to the so-called \emph{strong ordering} of the vertices. With this insight, MBT in bipartite permutation graphs reduces to finding a maximum crossing-free binary tree. We solve this latter problem by dynamic programming. 

\subsection{Organization} 
We present the $2^k \poly(n)$ time algorithm for \kBT in Section \ref{sec:arithmetization}. We present our hardness results for DAGs in Section \ref{sec:hardness-dags}. We formulate an IP for DAGs and discuss its integrality gap in Section \ref{sec:IP-gap-dags}. We show our hardness results for undirected graphs in Section \ref{sec:undirected-hardness}. 
We design an efficient algorithm for bipartite permutation graphs in Section \ref{sec:bipartite-permutation}. 
We conclude with a few open problems in Section \ref{sec:conclusions}.



%% file: preliminaries.tex
\subsection{Preliminaries}\label{sec:prelim}

\paragraph{PTAS and APX-hardness.}
We say that a maximization problem has a {\em polynomial-time approximation scheme} (PTAS) if it admits an algorithm that for each fixed $\eps>0$, and for each instance, outputs an $(1-\eps)$-approximate solution, in time polynomial in the size of the input instance. A problem is said to be in the class APX if it has a polynomial-time constant-factor approximation algorithm. A problem is APX-hard if there is a PTAS reduction from every problem in APX to that problem.

\paragraph{MBT in directed graphs.}
Given a directed graph $G=(V,E)$ and a vertex $r \in V$, we say that a subgraph $T$ where $V(T) \subseteq V$ and $E(T) \subseteq E$, is an \emph{$r$-rooted tree} in $G$ if $T$ is acyclic and every vertex $v$ in $T$ has a \emph{unique} directed path (in $T$) to $r$. If the in-degree of each vertex in $T$ is at most 2, then $T$ is an {\em $r$-rooted binary tree}. 

The problem of interest in directed graphs is the following:

\begin{problem}{\rdmaxBT}
	Given: A directed graph $G=(V, E)$ and a root $r \in V$.
	
	Goal: An $r$-rooted binary tree $T$ in $G$ with maximum number of vertices. 
\end{problem}

The problem \dagmaxBT is a special case of \rdmaxBT in which the input directed graph is a DAG. We recall that the rooted and unrooted variants of the maximum binary tree problem in DAGs are equivalent. 

\paragraph{MBT in undirected graphs.}
Given an undirected graph $G=(V,E)$, we say that a subgraph $T$, where $V(T) \subseteq V$ and $E(T) \subseteq E$, is a \emph{binary tree} in $G$ if $T$ is connected, acyclic, and $\deg_T(v) \le 3$ for every vertex $v \in V(T)$.
We will focus on the unrooted variant, i.e., \umaxBT, since the inapproximability results for the rooted variant are implied by inapproximability results for the unrooted variant.

\begin{problem}{\umaxBT}
	Given: An undirected graph $G$.
	
	Goal: A binary tree in $G$ with maximum number of vertices. 
\end{problem}

%% file: arithmetic-circuit.tex
\section{A $2^k \poly(n)$ time algorithm for \kBT}\label{sec:arithmetization}

In this section, we present a randomized algorithm that solves \kBT exactly and runs in time $2^k \poly(n)$ where $n$ is the number of vertices in the input graph. We recall that \kBT is the problem of deciding whether a given directed graph contains a binary tree of size $k$. Our algorithm is inspired by an algebraic approach for solving the \kpath problem---the algebraic approach relies on efficient detection of multilinear terms in a given  polynomial.\\

\noindent \textbf{\kpath, polynomials and multilinear terms.}
We begin with a recap of the algebraic approach to solve \kpath---here, the goal is to verify if a given (directed or undirected) graph $G$ contains a path of length at least $k$. 
There has been a rich line of research dedicated to designing algorithms for \kpath with running time $\beta^k \poly(n)$ where $\beta > 1$ is a constant and $n$ is the number of vertices in $G$ (cf. \cite{alon1995color, koutis2008faster, williams2009finding, bjorklund2017narrow}). In particular, the algorithms in \cite{koutis2008faster} and \cite{williams2009finding} are based on detecting multilinear terms in a polynomial. 

We now recall the problem of detecting multilinear terms in a polynomial. 
Here, we are given a polynomial with coefficients in a finite field $\GL{q}$ and the goal is to verify if it has a multilinear term. 
We emphasize that the input polynomial is given \emph{implicitly} by an arithmetic circuit consisting of additive and multiplicative gates. In other words, the algorithm is allowed to evaluate the polynomial at any point but does not have direct access to the sum-of-products expansion of the polynomial. 
We recall that a multilinear term in a polynomial $p \in \GL{q}[x_1, x_2, \cdots, x_m]$ is a monomial in the sum-of-products expansion of $p$ consisting of only degree-1 variables. For example, in the following polynomial
\begin{align*}
p(x_1, x_2, x_3) = x_1^2x_2 + x_3 + x_1x_2x_3, 
\end{align*}
the monomials $x_3$ and $x_1x_2x_3$ are multilinear terms, whereas $x_1^2x_2$ is not a multilinear term since $x_1$ has degree 2.
We will use the algorithm mentioned in the following theorem as a black box for detecting multilinear terms in a given polynomial. 

\begin{theorem}[Theorem 3.1 in \cite{williams2009finding}] \label{thm:test-multilinear}
Let $P(x_1, \cdots ,x_n)$ be a polynomial of degree at most $k$, represented by an arithmetic circuit of size $s(n)$ with additive gates (of unbounded fan-in), multiplicative gates (of fan-in two), and no scalar multiplications. There is a randomized algorithm that on input $P$ runs in $2^ks(n) \cdot \poly(n)\log\tp{1/
\delta}$ time, outputs `yes' with probability $1-\delta$ if there is a multilinear term in the sum-product expansion of $P$, and outputs `no' if there is no multilinear term.
\end{theorem}

The idea behind solving \kpath with the help of this theorem is to construct a polynomial $p_G$ based on the input graph $G$ so that $p_G$ contains a multilinear term if and only if $G$ contains a simple path of length $k$. At the same time, $p_G$ should be computable by an arithmetic circuit of size $\poly(n)$.
Koutis and Williams achieved these properties using the following polynomial: 
\begin{align*}
p_G(x_1, \cdots, x_n) := \sum_{\tp{v_{i_1}, v_{i_2}, \ldots, v_{i_k}}: \textup{ a walk in }G}x_{i_1}x_{i_2}\ldots x_{i_k}.
\end{align*}
We recall that a walk in $G$ is a sequence of vertices in which neighbouring vertices are adjacent in $G$. From the definition, it is easy to observe that there is a  one-to-one correspondence between simple $k$-paths in $G$ and multilinear terms in $p_G$. Moreover, it can be shown that there is an arithmetic circuit of size $O\tp{k^2(m+n)}$ that computes $p_G$, where $m$ is the number of edges and $n$ is the number of vertices in $G$. See Chapter 10.4 of
\cite{FPT-book} 
for alternative constructions of this polynomial.\\

\noindent \textbf{The polynomial construction for \kBT.}
Following the above-mentioned approach, we construct a polynomial $P_G$ with the property that $P_G$ contains a multilinear term if and only if $G$ contains a binary tree of size $k$. Unfortunately, there is no immediate generalization of walks of length $k$ 
that characterize binary trees on $k$ vertices. 
So, instead of defining the polynomial conceptually, we will define the polynomial recursively by building the arithmetic circuit that computes $P_G$, and will prove the correspondence between multilinear terms in $P_G$ and binary trees of size $k$ in $G$. In the definition of our polynomial, we also need to introduce an auxiliary variable to eliminate low-degree multilinear terms in $P_G$ (which is not an issue in the construction of the polynomial for \kpath).

Let $G=(V,E)$ be the given directed graph. For $v \in V$, let $\Delta^{in}_v \coloneqq \set{u \in V \colon (u,v) \in E}$. We begin by defining a polynomial $P_v\itn{k}$ for every $v \in V$ and every positive integer $k$, in $(n+1)$ variables $\set{x_v}_{v \in V} \cup \set{y}$:
{
\footnotesize
\begin{flalign*}
P_v\itn{k} \coloneqq \begin{cases}
x_v & \textup{if $k = 1$} \\
x_v \cdot y^{k-1} & \textup{if $k > 1$ and $\Delta^{in}_v = \varnothing$} \\
{\displaystyle x_v \tp{\sum_{u \in \Delta^{in}_v}P_u\itn{k-1} + \sum_{\ell=1}^{k-2}\tp{\sum_{u_1 \in \Delta^{in}_v}P_{u_1}\itn{\ell}}\tp{\sum_{u_2 \in \Delta^{in}_v}P_{u_2}\itn{k-1-\ell}} } } & \textup{if $k > 1$ and $\Delta^{in}_v \neq \varnothing$}
\end{cases}
\end{flalign*}
}

Next, we define $P_G\itn{k} \coloneqq \sum_{v \in V}P_v\itn{k}$. We recall that a polynomial is homogenous if every monomial has the same degree. By induction on $k$, the polynomial $P_v\itn{k}$ is a degree-$k$ homogeneous polynomial and so is $P_G\itn{k}$. Moreover, by the recursive definition, we see that $P_v\itn{k}$ can be represented as an arithmetic circuit of size $O(k^2n)$ since there are $kn$ polynomials in total, and computing each requires $O(1)$ addition gates (with unbounded fan-in) and $O(k)$ multiplication gates (with fan-in two). 
We show the following connection between multilinear terms in $P_G\itn{k}$ and binary trees in $G$. 

\begin{restatable}{lemma}{kbttomultilinear} \label{lem:kBT-to-multilinear}
The graph $G$ has a binary tree of size $k$ rooted at $r$ if and only if there is a multilinear term of the form $\prod_{v\in S}x_v$ in $P_r\itn{k}$ where $|S|=k$.
\end{restatable}
\begin{proof}
We first show the forward direction, i.e., if $G$ has a binary tree $T$ of size $k$ rooted at $r$, then there is a multilinear term of the form $\prod_{v\in T}x_v$ in $P_r\itn{k}$. We prove this by induction on $k$. 
The base case $k=1$ follows since $P_r\itn{1} = x_r$.
Suppose that the forward direction holds when $|T|\le k-1$. For $|T|=k$, we consider two cases.
\begin{enumerate}
\item The root $r$ has only one child $c$. The subtree $T_c$ of $T$ rooted at $c$ has size $k-1$. By induction hypothesis there is a multilinear term $\prod_{v \in T_c}x_v$ in $P_c\itn{k-1}$. Since $c \in \Delta^{in}_r$, for some polynomial $Q$ we can write 
\begin{align*}
P_r\itn{k} = x_r \tp{P_c\itn{k-1} + Q}.
\end{align*}
Therefore $x_r \cdot \prod_{v \in T_c}x_v$ is a term in $P_r\itn{k}$. This term is multilinear and equals to $\prod_{v \in T}x_v$ since $r \notin T_c$.

\item The root $r$ has two children $c_1, c_2$. Suppose that the subtree $T_{c_1}$ rooted at $c_1$ has size $\ell$, thus the subtree $T_{c_2}$ rooted at $c_2$ has size $k-1-\ell$. The induction hypothesis implies that $P_{c_1}\itn{\ell}$ has a multilinear term $\prod_{v \in T_{c_1}}x_v$, and $P_{c_2}\itn{k-1-\ell}$ has a multilinear term $\prod_{v \in T_{c_2}}x_v$. Since $c_1, c_2 \in \Delta^{in}_r$, for some polynomial $Q$ we can write
\begin{align*}
P_r\itn{k} = x_r \tp{P_{c_1}\itn{\ell}P_{c_2}\itn{k-1-\ell} + Q}.
\end{align*}
Therefore $x_r\tp{\prod_{v \in T_{c_1}}x_v}\tp{\prod_{v \in T_{c_2}}x_v}$ is a term in $P_r\itn{k}$. This term is multilinear and equals to $\prod_{v \in T}x_v$ because $T$ is the disjoint union of $r$, $T_{c_1}$ and $T_{c_2}$.
\end{enumerate}  
In both cases, the polynomial $P_r\itn{k}$ has a multilinear term $\prod_{v \in T}x_v$. This completes the inductive step.

Next, we show that if $P_r\itn{k}$ has a multilinear term of the form $\prod_{v \in S}x_v$ where $|S|=k$, then there is a binary tree $T$ rooted at $r$ in $G$ with vertex set $S$. We prove this also by induction on $k$. The base case $k=1$ is trivial since $P_r\itn{1} = x_r$ and there is a binary tree of size 1 rooted at $r$. Suppose that the statement holds for $k-1$ or less ($k > 1$).

Let $\prod_{v \in S}x_v$ be a multilinear term in $P_r\itn{k}$. We note that $r \in S$ since every term in $P_r\itn{k}$ contains $x_r$. Moreover, we may assume that $\Delta^{in}_r \neq \varnothing$ since otherwise $P_r\itn{k} = x_r \cdot y^{k-1}$ which does not contain any term of the form $\prod_{v \in S}x_v$. According to the definition of $P_r\itn{k}$, we could have two cases.
\begin{enumerate}
\item The term $\prod_{v \in S \setminus \set{r}}x_v$ is a multilinear term in $P_c\itn{k-1}$ for some $c \in \Delta^{in}_r$. The induction hypothesis implies that there is a binary tree $T_c$ rooted at $c$ with vertex set $S \setminus \set{r}$. Let $T$ be the binary tree obtained by adding the edge $(c, r)$ to $T_c$. Then $T$ is a binary tree rooted at $r$ with vertex set $S$.
\item The term $\prod_{v \in S \setminus \set{r}}x_v$ is a multilinear term in $P_{c_1}\itn{\ell}P_{c_2}\itn{k-1-\ell}$ for some $c_1, c_2 \in \Delta^{in}_r$ and some integer $1 \le \ell \le k-2$. In this case, since $P_{c_1}\itn{\ell}$ and $P_{c_2}\itn{k-1-\ell}$ are homogeneous polynomials of degree $\ell$ and $k-1-\ell$, we can partition $S \setminus \set{r}$ into two sets $S_1$ and $S_2$ with $|S_1|=\ell$ and $|S_2|=k-1-\ell$ such that $\prod_{v \in S_1}x_v$ is a multilinear term in $P_{c_1}\itn{\ell}$, and $\prod_{v \in S_2}x_v$ is a multilinear term in $P_{c_2}\itn{\ell}$. Applying the induction hypothesis, we obtain a binary tree $T_{c_1}$ (rooted at $c_1$) with vertex set $S_1$ and a binary tree $T_{c_2}$ (rooted at $c_2$) with vertex set $S_2$. Let $T$ be the binary tree obtained by adding edges $(c_1, r)$ and $(c_2, r)$ to $T_{c_1} \cup T_{c_2}$. Then $T$ is a binary tree rooted at $r$ with vertex set $S_1 \cup S_2 \cup \set{r} = S$.
\end{enumerate}
In both cases, we can find a binary tree $T$ rooted at $r$ with vertex set $S$. This completes the inductive step.
\end{proof}

With this choice of $P_G\itn{k}$, we call the algorithm appearing in Theorem \ref{thm:test-multilinear} on input polynomial $\tilde{P}_G\itn{k} \coloneqq y \cdot P_G\itn{k}$, and output the result. We note that every multilinear term of the form $\prod_{v \in S}x_v$ in $P_G\itn{k}$ becomes a multilinear term of the form $y \cdot \prod_{v \in S}x_v$ in $\tilde{P}_G\itn{k}$, and every multilinear term of the form $y \cdot \prod_{v \in S}x_v$ in $P_G\itn{k}$ becomes $y^2 \cdot \prod_{v \in S}x_v$ in $\tilde{P}_G\itn{k}$, which is no longer a multilinear term. In light of Lemma \ref{lem:kBT-to-multilinear}, 
the graph $G$ contains a binary tree of size $k$ if and only if the degree-$(k+1)$ homogeneous polynomial $\tilde{P}_G\itn{k}$ has a multilinear term. The running time is $2^{k+1} \cdot O(k^2n) \cdot \poly(n+1)\log\tp{1/\delta} = 2^k \cdot \poly(n)\log\tp{1/\delta}$.

We remark that this algorithm does not immediately tell us the tree $T$ (namely the edges in $T$). However, we can find the edges in $T$ with high probability via a reduction from the search variant to the decision variant. This is formalized in the next lemma.

\begin{restatable}{lemma}{ProbAmp}\label{lemma:probability-amplification}
Suppose that there is an algorithm $\+A$ which takes as input a directed graph $G=(V,E)$, an integer $k$ and $\delta' \in (0,1)$ runs in time $2^k \poly\tp{|V|} \log\tp{1/\delta'}$ and
\begin{itemize}
\item outputs 'yes' with probability at least $1-\delta'$ if $G$ contains a binary tree of size $k$, 
\item outputs 'no' with probability 1 if $G$ does not contain a binary tree of size $k$. 
\end{itemize}
Then there also exists an algorithm $\+A'$ which for every $\delta \in (0,1)$ outputs a binary tree $T$ of size $k$ with probability at least $1-\delta$ when the answer is 'yes', and runs in time $2^k \poly\tp{|V|}\log\tp{1/\delta}$.
\end{restatable}


\begin{proof}
The algorithm $\+A'$ iterates through all arcs $e \in E$ and calls $\+A$ on $(G-e, k)$ with $\delta' = \delta/m$ where $G-e=(V, E \setminus \set{e})$ and $m = |E|$. If for some $e \in E$ the call to $\+A$ outputs 'yes', we remove $e$ from $G$ (i.e., set $G \gets G-e$) and continue the process. We will show that when the algorithm terminates, the arcs in $G$ constitute a binary tree of size $k$ (if there exists one) with probability at least $1-\delta$.

Suppose the order in which $\+A'$ processes the arcs is $e_1, e_2, \cdots, e_m$, and the graph at iteration $t$ is denoted by $G\itn{t}$. Let $B_t$ denote the event ``$G\itn{t-1}-e_t$ contains a binary tree of size $k$, but the call to $\+A\tp{G\itn{t-1}-e_t, k}$ returns no''. Due to the assumption we made for $\+A$, event $B_t$ happens with probability at most $\delta'$. Since the algorithm $\+A$ has perfect soundness, whenever $\+A'$ removes an edge we are certain that the remaining graph still contains a binary tree of size $k$ (otherwise the call to $\+A$ would never return `yes'). That means if $G\itn{0}=G$ contains a binary tree of size $k$ then $G\itn{t}$ contains a binary tree of size $k$ for all $0 \le t \le m$. Therefore if none of the events $B_t$ happens, the final graph $G\itn{m}$ is a binary tree of size $k$. The probability of failure is upper bounded by 
\begin{align*}
\Pr\left[\bigcup_{t=1}^{m}B_t\right] \le m \cdot \delta' = m \cdot \frac{\delta}{m} = \delta.
\end{align*} 
Since algorithm $\+A'$ makes $m$ calls to algorithm $\+A$, the running time of $\+A'$ is $m \cdot 2^k \poly\tp{|V|} \log\tp{1/\delta'} = 2^k \poly\tp{|V|}\log\tp{1/\delta}$.  
\end{proof}

Theorem~\ref{thm:test-multilinear} in conjunction with Lemmas~\ref{lem:kBT-to-multilinear} and \ref{lemma:probability-amplification} complete the proof of Theorem~\ref{theorem:k-binary-tree}.

%% file: dags-hardness.tex
\section{Hardness results for DAGs}\label{sec:hardness-dags}

In this section, we show the inapproximability of finding a maximum binary tree in DAGs. The \emph{size} of a binary tree denotes the number of vertices in the tree.\\ 

\subsection{Self-improvability for directed graphs} \label{sec:dag-self-improvability}

We show that an algorithm for \rdmaxBT achieving a constant factor approximation can be used to design a PTAS in Theorem \ref{thm:d-ptas}. We emphasize that this result holds for arbitrary directed graphs and not just DAGs. The idea is to gradually boost up the approximation ratio by running the constant-factor approximation algorithm on squared graphs. 
Our notion of \emph{squared graph} will be the following.

\begin{definition} \label{def:d-squared}
Given a directed graph $G=(V,E)$ with root $r$, the \emph{squared graph} $G^2$ is the directed graph obtained by performing the following operations on $G$: 
\begin{enumerate}
\item Construct $G'=(V',E')$ by introducing a source vertex $s$, i.e., $V':=V \cup \{s\}$. We add arcs from $s$ to every vertex in $G$, i.e., $E' := E \cup \{\tp{s,v} \colon v \in V\}$.
\item For each $u \in V$ (we note that $V$ does not include the source vertex), we create a copy of $G'$ that we denote as a {\em vertex copy} $G_u'$. We will denote the root vertex of $G_u'$ by $r_u$, and the source vertex of $G_u'$ by $s_u$. 
\item For each $\tp{u,v} \in E$, we create an arc $\tp{r_u, s_v}$.
\item We declare the root of $G^2$ to be $r_r$, i.e. the root vertex of the vertex copy $G_r'$.
\end{enumerate}
We define $G^{2^{k+1}}$ recursively as $G^{2^{k+1}} \coloneqq \tp{G^{2^k}}^2$ with the base case $G^1 \coloneqq G$.
\end{definition}

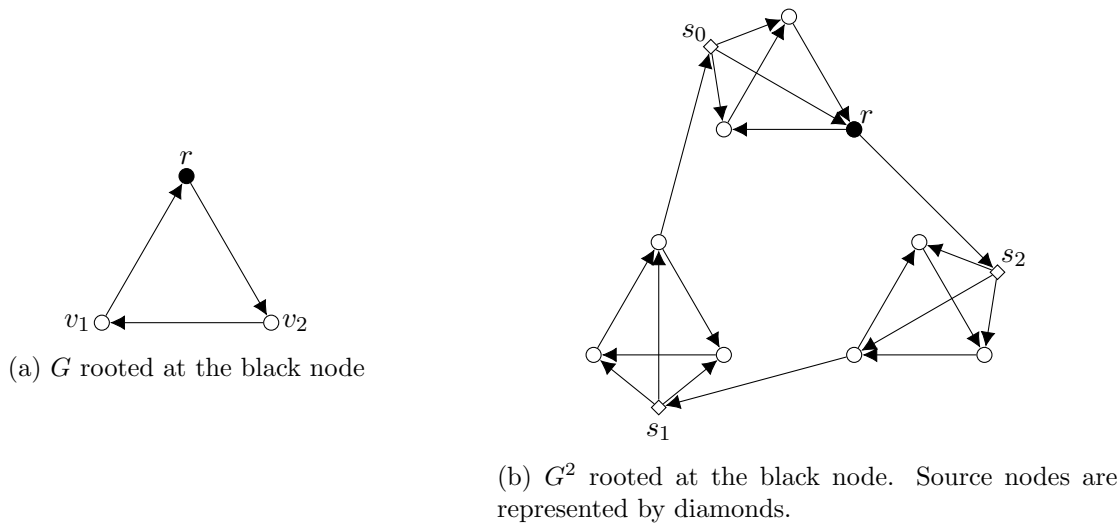
\begin{figure}[h]
\centering
\begin{subfigure}{.5\textwidth}
\centering
\begin{tikzpicture}[every node/.style={draw=black, circle, inner sep=0pt, minimum size=0.2cm}]
\def \r {1.3cm}
\node[label={$r$},fill] (0) at (90:\r) {};
\node[label={left:$v_1$}] (1) at (90+120*1:\r) {};
\node[label={right:$v_2$}] (2) at (90+120*2:\r) {};
\path[->] 
    (0) edge (2) 
    (2) edge (1) 
    (1) edge (0)
    ;
\end{tikzpicture}
\caption{$G$ rooted at the black node}
\end{subfigure}%
\begin{subfigure}{.5\textwidth}
\centering
\begin{tikzpicture}[every node/.style={draw=black, circle, inner sep=0pt, minimum size=0.2cm}]
\def \r {2cm}
\def \l {1cm}
\foreach \i in {0,1,2}{
    \tikzset{shift={(90+120*\i:\r)}}
    \node[label=150+120*\i:$s_\i$, diamond] (t\i) at (150+120*\i:1.2*\l) {};
    \foreach \j in {0,1,2}
        \node (\i\j) at (90+120*\j:\l) {};
}
\tikzset{shift={(90:\r)}}
\node[label=45:$r$, fill] (02) at (330:\l) {};
\foreach \i in {0,1,2}{
    \path[->]
        (\i0) edge (\i2)
        (\i2) edge (\i1)
        (\i1) edge (\i0)
        ;
    \foreach \j in {0,1,2}
        \path[->] (t\i) edge (\i\j);
}
\path[->]
    (02) edge (t2)
    (21) edge (t1)
    (10) edge (t0)
    ;
\end{tikzpicture}
\caption{$G^2$ rooted at the black node. Source nodes are represented by diamonds.}
\end{subfigure}
\caption{Directed Squared Graph}
\end{figure}

Given a directed graph $G$ with $n-1$ vertices, the number of vertices in $G^{2^k}$ satisfies the recurrence relation 
\[\abs{V\tp{G^{2^k}}} = \abs{V\tp{G^{2^{k-1}}}} \cdot \tp{\abs{V\tp{G^{2^{k-1}}}} + 1} = \abs{V\tp{G^{2^{k-1}}}}^2 + \abs{V\tp{G^{2^{k-1}}}}.\]
Hence, we have 
\[\abs{V\tp{G^{2^k}}} + 1 \le \tp{\abs{V\tp{G^{2^{k-1}}}} + 1}^2 \le \tp{\abs{V\tp{G^{2^{k-2}}}} + 1}^{2^2} \le \cdots \le \tp{\abs{V\tp{G^{2^0}}} + 1}^{2^k} = n^{2^k}.\]

We use $OPT(G)$ to denote the size (number of vertices) of a maximum binary tree in $G$. The following lemma shows that $OPT(G)$ is super-multiplicative under the squaring operation.

\begin{lemma} \label{lem:di-squared-opt}
For any fixed root $r$, $OPT(G^2)\ge OPT(G)^2$.
\end{lemma}
\begin{proof}
Suppose we have an optimal $r$-rooted binary tree $T_1$ of $G$, i.e. $|V\tp{T_1}|=OPT(G)$.  We construct an $r_r$-rooted binary tree $T_2$ of $G^2$ as follows:
\begin{enumerate}
\item For $v \in V\tp{G}$, define $T_v' = T_v \cup \set{s_v}$ to be the optimal $r_v$-rooted binary tree in the vertex copy $G_v'$ where $T_v$ is identical to $T_1$ and the source vertex $s_v$ is connected to an arbitrary leaf node in $T_v$.
\item For every vertex $v \in T_1$, add $T_v'$ to $T_2$. This step generates $|V\tp{T_1}| \cdot \tp{|V\tp{T_1}|+1}$ vertices in $T_2$.
\item Connect the copies selected in step 2 by adding the arc $\tp{r_u, s_v}$ to $T_2$ for every arc $\tp{u,v} \in T_1$. 
\end{enumerate}
Since $T_1$ is an $r$-rooted binary tree (in $G$), it follows that $T_2$ is an $r_r$-rooted binary tree (in $G^2$). Moreover, the size of $T_2$ is 
\begin{align*}
|V\tp{T_2}| = |V\tp{T_1}| \cdot \tp{|V\tp{T_1}|+1} \ge OPT\tp{G}^2,
\end{align*} 
which cannot exceed $OPT\tp{G^2}$.
\end{proof}

The following lemma shows that a large binary tree in $G^2$ can be used to obtain a large binary tree in $G$. 
\begin{lemma} \label{lem:di-sqrtfactor}
For every $\alpha \in (0,1]$, given an $r_r$-rooted binary tree $T_2$ in $G^{2}$ with size
\begin{align*}
|V\tp{T_2}| \ge \alpha OPT\tp{G^{2}} - 1,
\end{align*}
there is a linear-time (in the size of $G^{2}$) algorithm that finds an $r$-rooted binary tree $T_1$ of $G$ with size 
\begin{align*}
|V\tp{T_1}| \ge \sqrt{\alpha} OPT\tp{G} - 1.
\end{align*}
\end{lemma}
\begin{proof}
Let $U \coloneqq \{v \colon v \in V(G) \text{ such that } r_v \in V(T_2)\}$ 
and $A \coloneqq \{\tp{v,w} \colon v, w \in V(G), \tp{r_v, s_w} \in E(T_2)\}$. We note that $T_1':=(U,A)$ is an $r$-rooted binary tree in $G$. This is because the path from every $v \in U$ to the root $r$ is preserved, and the in-degree of every node $w \in U$ is bounded by the in-degree of $s_w$ (in $T_2$), which is thus at most $2$, and similarly the out-degree of every node is at most $1$. We also remark that $T_1'$ can be found in linear time. If $|U| \ge \sqrt{\alpha}OPT(G) > \sqrt{\alpha}OPT(G)-1$, then the lemma is already proved. So, we may assume that $|U|<\sqrt{\alpha}OPT(G)$. 

We now consider $T_v' \coloneqq \tp{V\tp{T_2} \cap V\tp{G_v'}, E\tp{T_2} \cap E\tp{G_v'}}$ for $v \in U$. We can view $T_v'$ as the restriction of $T_2$ to $G_v'$, hence every node of $T_v'$ has out-degree at most 2. Since $T_2$ is an $r_r$-rooted binary tree in $G^2$, every vertex in $V\tp{T_2} \cap V\tp{G_v'}$ has a unique directed path (in $T_2$) to $r_r$, which must go through $r_v$, thus every vertex in $V\tp{T_2} \cap V\tp{G_v'}$ has a unique directed path to $r_v$. It follows that $T_v'$ is an $r_v$-rooted binary tree in the vertex copy $G_v'$.

We now show that there exists $v \in U$ such that $|V\tp{T_v'}|\ge \sqrt{\alpha}OPT(G)$. 
Suppose not, which means for every $v \in U$ we have $|V\tp{T_v'}|<\sqrt{\alpha}OPT(G)$. Then  
\begin{align*}
|V\tp{T_2}| &=\sum_{v\in U}|V\tp{T_v'}|< \sum_{v \in U}\tp{\sqrt{\alpha} OPT(G)} < \sqrt{\alpha} OPT(G) \cdot \sqrt{\alpha}OPT(G) \\
 &= \alpha OPT(G)^2 \le \alpha \cdot OPT\tp{G^2},
\end{align*}
a contradiction. The last inequality is due to Lemma \ref{lem:di-squared-opt}.

In linear time we can find a binary tree $T_v'$ with the desired size $|V\tp{T_v'}| \ge \sqrt{\alpha}OPT(G)$. To complete the proof of the lemma, we let $T_1 \coloneqq T_v' \setminus \set{s_v}$ which is (isomorphic to) an $r$-rooted binary tree in $G$ with size at least $\sqrt{\alpha}OPT(G) - 1$.
\end{proof}

\begin{restatable}{theorem}{DagDPtas}\label{thm:d-ptas}
If \rdmaxBT has a polynomial-time algorithm that achieves a constant-factor approximation, then it has a PTAS.
\end{restatable}

\begin{proof}
Suppose that we have a polynomial-time algorithm $\+A$ that achieves an $\alpha$-approximation for \rdmaxBT. Given a directed graph $G$, root $r$ and $\eps > 0$, let
\begin{align*}
k \coloneqq 1 + \ceilfit{ \log_2\frac{\log_2 \alpha}{\log_2 (1-\eps)} }
\end{align*}
be an integer constant that depends on $\alpha$ and $\eps$. We construct $G^{2^k}$ and run algorithm $\+A$ on $G^{2^k}$. Then, we get a binary tree in $G^{2^k}$ of size at least $\alpha OPT\tp{G^{2^k}} - 1$. By Lemma \ref{lem:di-sqrtfactor}, we can obtain an $r$-rooted binary tree in $G$ of size at least  
\begin{align*}
\alpha^{2^{-k}} OPT(G) - 1 \ge \alpha^{2^{-k+1}} OPT(G) \ge (1-\eps) OPT(G). 
\end{align*}
The first inequality holds as long as
\begin{align*}
OPT(G) \ge \frac{1}{\sqrt{1-\eps} - (1-\eps)} \ge \frac{1}{\alpha^{2^{-k}} - \alpha^{2^{-k+1}}}.
\end{align*}
We note that if $OPT(G)$ is smaller than $1/\tp{\alpha^{2^{-k}} - \alpha^{2^{-k+1}}}$ which is a constant, then we can solve the problem exactly by brute force in polynomial time. Finally, we also observe that for fixed $\eps$, the running time of this algorithm is polynomial since there are at most $n^{2^k}=n^{O(1)}$ vertices in the graph $G^{2^k}$.
\end{proof}

\subsection{APX-hardness for DAGs}
Next, we show the inapproximability results for DAGs. 
We begin by recalling \dagmaxBT: 
We begin by recalling the problem: 
\begin{problem}{\dagmaxBT}
	Given: A directed acyclic graph $G=(V, E)$ and a root $r \in V$.
	
	Goal: An $r$-rooted binary tree in $G$ with the largest number of nodes. 
\end{problem}
We may assume that the root is the only vertex that has no outgoing arcs as we may discard all vertices that cannot reach the root. We show that \dagmaxBT is APX-hard by reducing from the following problem. 


\begin{problem}{\textsc{Max-$3$-Colorable-Subgraph}}
Given: An undirected graph $G$ that is $3$-colorable.

Goal: A 3-coloring of $G$ that maximizes the fraction of properly colored edges.
\end{problem}

It is known that finding a $3$-coloring that properly colors at least $32/33$-fraction of edges in a given $3$-colorable graph is NP-hard \cite{guruswami2013improved, AOW12}. In particular, \textsc{Max-$3$-Colorable-Subgraph} is APX-hard. We reduce \textsc{Max-$3$-Colorable-Subgraph} to \dagmaxBT. Let $G=(V,E)$ be the input $3$-colorable undirected graph with $n:=|V|$ and $m:=|E|$. For $\eps > 0$ to be fixed later, we construct a DAG, denoted $D(G, \eps)$, as follows (see Figure \ref{fig:dag-hardness-reduction} for an illustration):

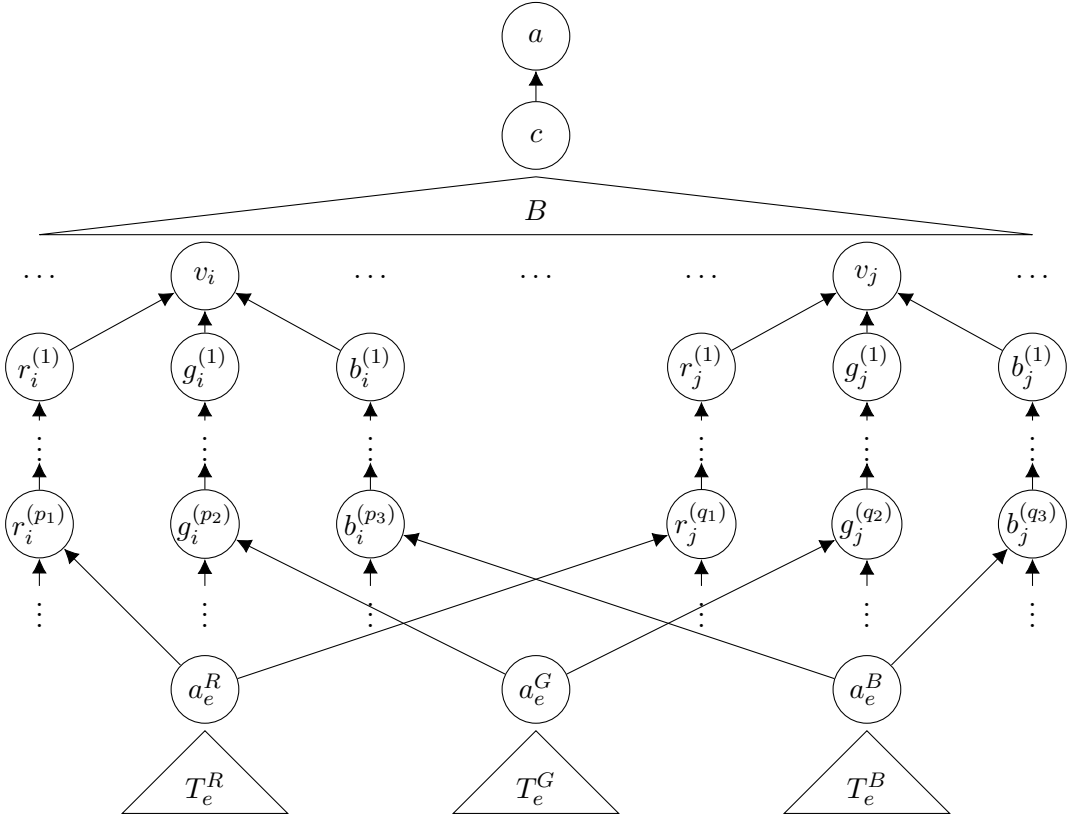
\begin{figure}[h]
\centering
\begin{tikzpicture}[every node/.style={draw=black, circle, inner sep=0pt, minimum size=0.9cm},scale=1.1]
	\def \l {1}
	\node[] (a) at (0,2.9) {$a$};
	\node[] (c) at (0,1.7) {$c$};
	\path[->] (c) edge (a);
	
	\node[](Vi)	at (-4, 0)	{$v_i$};
	\node[](Ri1)	at (-6, -1.1)		{$r_i\itn{1}$};
	\node[draw=none, minimum size=0.5cm](Ri2)	at (-6, {-1 - \l})		{$\vdots$};
	\node[](Ri3)	at (-6, {-1 - 2*\l})		{$r_i\itn{p_1}$};
	\node[draw=none, minimum size=0.5cm](Ri4)	at (-6, {-1 - 3*\l})		{$\vdots$};
	\node[](Gi1)	at (-4, -1.1)	{$g_i\itn{1}$};
	\node[draw=none, minimum size=0.5cm](Gi2)	at (-4, {-1 - \l})	{$\vdots$};
	\node[](Gi3)	at (-4, {-1 - 2*\l})		{$g_i\itn{p_2}$};
	\node[draw=none, minimum size=0.5cm](Gi4)	at (-4, {-1 - 3*\l})		{$\vdots$};
	\node[](Bi1)	at (-2, -1.1)	{$b_i\itn{1}$};
	\node[draw=none, minimum size=0.5cm](Bi2)	at (-2, {-1 - \l})	{$\vdots$};
	\node[](Bi3)	at (-2, {-1 - 2*\l})		{$b_i\itn{p_3}$};
	\node[draw=none, minimum size=0.5cm](Bi4)	at (-2, {-1 - 3*\l})		{$\vdots$};

	\node[](Vj)	at (4, 0)	{$v_j$};
	\node[](Rj1)	at (2, -1.1)		{$r_j\itn{1}$};
	\node[draw=none, minimum size=0.5cm](Rj2)	at (2, {-1 - \l})		{$\vdots$};
	\node[](Rj3)	at (2, {-1 - 2*\l})		{$r_j\itn{q_1}$};
	\node[draw=none, minimum size=0.5cm](Rj4)	at (2, {-1 - 3*\l})		{$\vdots$};
	\node[](Gj1)	at (4, -1.1)	{$g_j\itn{1}$};
	\node[draw=none, minimum size=0.5cm](Gj2)	at (4, {-1 - \l})	{$\vdots$};
	\node[](Gj3)	at (4, {-1 - 2*\l})		{$g_j\itn{q_2}$};
	\node[draw=none, minimum size=0.5cm](Gj4)	at (4, {-1 - 3*\l})		{$\vdots$};
	\node[](Bj1)	at (6, -1.1)	{$b_j\itn{1}$};
	\node[draw=none, minimum size=0.5cm](Bj2)	at (6, {-1 - \l})	{$\vdots$};
	\node[](Bj3)	at (6, {-1 - 2*\l})		{$b_j\itn{q_3}$};
	\node[draw=none, minimum size=0.5cm](Bj4)	at (6, {-1 - 3*\l})		{$\vdots$};
	
	\node[draw=none, minimum size=0.5cm](dots)	at (-2,0)		{$\hdots$};
	\node[draw=none, minimum size=0.5cm](dots)	at (0,0)		{$\hdots$};
	\node[draw=none, minimum size=0.5cm](dots)	at (2,0)		{$\hdots$};
	
	\node[] (aeR) at (-4,-5) {$a_e^R$};
	\draw (-5,-6.5) -- (-3,-6.5)-- (-4,-5.5) --cycle;
	\node[draw=none] at (-4,-6.2) {$T_e^R$};
	\node[] (aeG) at (0,-5) {$a_e^G$};
	\draw (-1,-6.5) -- (1,-6.5)-- (0,-5.5) --cycle;
	\node[draw=none] at (0,-6.2) {$T_e^G$};
	\node[] (aeB) at (4,-5) {$a_e^B$};
	\draw (3,-6.5) -- (5,-6.5)-- (4,-5.5) --cycle;
	\node[draw=none] at (4,-6.2) {$T_e^B$};
	
	\path[->]
	(Ri1) edge (Vi)
	(Gi1) edge (Vi)
	(Bi1) edge (Vi)
	(Ri2) edge (Ri1)
	(Ri3) edge (Ri2)
	(Ri4) edge (Ri3)
	(Gi2) edge (Gi1)
	(Gi3) edge (Gi2)
	(Gi4) edge (Gi3)
	(Bi2) edge (Bi1)
	(Bi3) edge (Bi2)
	(Bi4) edge (Bi3)
	
	(Rj1) edge (Vj)
	(Gj1) edge (Vj)
	(Bj1) edge (Vj)
	(Rj2) edge (Rj1)
	(Rj3) edge (Rj2)
	(Rj4) edge (Rj3)
	(Gj2) edge (Gj1)
	(Gj3) edge (Gj2)
	(Gj4) edge (Gj3)
	(Bj2) edge (Bj1)
	(Bj3) edge (Bj2)
	(Bj4) edge (Bj3)
	
	(aeR) edge (Ri3)
	(aeR) edge (Rj3)
	(aeG) edge (Gi3)
	(aeG) edge (Gj3)
	(aeB) edge (Bi3)
	(aeB) edge (Bj3)
	;
	
	\draw (-6,0.5) -- (6,0.5) -- (0,1.2) --cycle;
    \node[draw=none](dots)	at (-6,0)		{$\hdots$};
    \node[draw=none](dots)	at (6,0)		{$\hdots$};
    \node[draw=none] at (0,0.8) {$B$};
	
\end{tikzpicture}

\setlength{\belowcaptionskip}{-15pt}
\caption{DAG $D(G,\eps)$ constructed in the reduction from \textsc{Max-$3$-Colorable-Subgraph} to \dagmaxBT.}
\label{fig:dag-hardness-reduction}
\end{figure}

\begin{enumerate}
\item Create a directed binary tree $B$ rooted at node $c$ with $n\coloneqq|V|$ leaf nodes. We will identify each leaf node by a unique vertex $v \in V$. Create a super root $a$ and arc $c \rightarrow a$. This tree and the super root would have $2n$ nodes, including the super root node $a$, $n$ leaf nodes, and $n-1$ internal nodes.

\item For every $i \in V$, we introduce three directed paths of length $n$ that will be referred to as  $R_i, G_i$ and $B_i$. Let $R_i$ be structured as $r_i\itn{1} \leftarrow r_i\itn{2} \leftarrow \cdots \leftarrow r_i\itn{n}$, and similarly introduce $g_i\itn{k}$ and $b_i\itn{k}$ with the same structure. Also add arcs $r_i\itn{1} \rightarrow v_i$, $g_i\itn{1} \rightarrow v_i$ and $b_i\itn{1} \rightarrow v_i$.

\item For every edge $e=\set{i, j} \in E$, introduce three directed binary trees that will be referred to as $T_e^R, T_e^G$, and $T_e^B$, each with
$t = \ceilfit{ \frac{2\eps n(n+1) +4n^2}{\eps m} } $
nodes. Let the roots of the binary trees $T_e^R, T_e^G$, and $T_e^B$ be $a_e^R, a_e^G$, and $a_e^B$ respectively. Add arcs $a_e^R \rightarrow r_i\itn{p_1}$ and $a_e^R \rightarrow r_j\itn{q_1}$ where $r_i\itn{p_1}$ and $r_j\itn{q_1}$ are two nodes in $R_i$ and $R_j$ with in-degree strictly smaller than 2. We note that $R_i$ is a path with $n$ nodes so such a node always exists. Similarly connect $a_e^G$ to $g_i\itn{p_2}$ and $g_j\itn{q_2}$, and $a_e^B$ to $b_i\itn{p_3}$ and $b_j\itn{q_3}$ in the directed paths $G_i$ and $B_i$, respectively.
\end{enumerate}

The constructed graph $D(G,\eps)$ is a DAG. We fix $a$ to be the root. The number of nodes  $N$ in $D(G, \eps)$ is $N = 3mt + 3n\cdot n + 2n = 3mt + 3n^2 + 2n$.
We note that every node $v_i\in V$ has in-degree exactly $2$ in every $a$-rooted maximal binary tree in $D(G, \eps)$. The idea of this reduction is to encode the color of $v_i$ as the unique path among $R_i, G_i, B_i$ that is \emph{not} in the subtree under $v_i$. The following two lemmas summarize the main properties of the DAG constructed above.

\begin{restatable}{lemma}{LbBigtree}\label{lem:lb-bigtree}
Let $T$ be a maximal $a$-rooted binary tree of $D(G,\eps)$. If $|V(T)| \ge (1-\eps/4)(N-n^2)$, then at most $\eps m$ nodes among $\cup_{e\in E}\{a_e^R,a_e^G,a_e^B\}$ are not in $T$.
\end{restatable}
\begin{proof}
Suppose more than $\eps m$ such nodes are missing from $T$. For each node $a_e^R$ that is not in $T$, the corresponding subtree $T_e^R$ is also not in $T$ (same for $a_e^G$ and $a_e^B$). Therefore
\begin{align*}
|V(T)| < N - \eps mt = 3mt + 3n^2 + 2n - \eps mt = \tp{1-\frac{\eps}{4}}\cdot 3mt + 3n^2 + 2n - \frac{\eps}{4}mt.
\end{align*}
The choice of $t$ implies that $\eps mt/4 > \eps n(n+1)/2 + n^2$. Therefore
\begin{align*}
|V(T)| &< \tp{1-\frac{\eps}{4}}\cdot 3mt + 3n^2 + 2n - \frac{\eps n(n+1)}{2} - n^2\\
 &< \tp{1-\frac{\eps}{4}}\cdot 3mt + \tp{1-\frac{\eps}{4}}\tp{2n^2 + 2n} \\
 &= \tp{1-\frac{\eps}{4}}(N-n^2),
\end{align*}
a contradiction. 
\end{proof}

\begin{restatable}{lemma}{OptBds}\label{lem:optbds}
If $G$ is 3-colorable, then every $a$-rooted maximum binary tree in $D(G,\eps)$ has size exactly $N-n^2$.
\end{restatable}
\begin{proof}
We first note that every binary subtree of $D(G,\eps)$ has size at most $N-n^2$. This is because there are $n$ vertices with in-degree $3$ (namely $v_1, v_2, \cdots, v_n$). For each such vertex $v_i$, there are $3$ vertices $r_i\itn{1}, g_i\itn{1}$ and $b_i\itn{1}$ whose only outgoing arc is to $v_i$. Moreover, each vertex $r_i\itn{1}$ (and similarly $g_i\itn{1}$ and $b_i\itn{1}$) is the end-vertex of an induced path of length $n$. 

Suppose $G$ is 3-colorable. We now construct an $a$-rooted binary tree $T$ of size $N-n^2$ in $D(G,\eps)$. We focus on the nodes to be discarded so that we may construct a binary spanning tree with the remaining nodes. Let $\sigma \colon V \rightarrow \set{Red, Green, Blue}$ be a proper $3$-coloring of $G$. If $\sigma(v_i) = Red$, we discard the path $R_i$. The cases where $\sigma(v_i) \in \set{Green, Blue}$ are similar. Since there are no monochromatic edges, there do not exist $e=\set{v_i, v_j} \in E$ and $C \in \set{R,G,B}$ such that both parents of $a_e^C$ are not in $T$. Therefore every binary tree $T_e^C$ is contained as a subtree in $T$.
\end{proof}


\begin{restatable}{theorem}{ThmDagNoPtas}\label{theorem:dagmaxBT-no-ptas}
Suppose there is a PTAS for \dagmaxBT on DAGs, then for every $\eps>0$ there is a polynomial-time algorithm which takes as input an undirected 3-colorable graph $G$, and outputs a 3-coloring of $G$ that properly colors at least $(1-\eps)m$ edges.
\end{restatable}
\begin{proof}
Let $G=(V,E)$ be the given undirected 3-colorable graph. We construct $D(G, \eps)$ in polynomial time. We note that the constructed graph $D(G,\eps)$ is a directed acyclic graph. We now run the PTAS for \dagmaxBT on $D(G,\eps)$ and root $a$ to obtain a $(1-\eps/4)$-approximate maximum binary tree in $D(G,\eps)$. By Lemma \ref{lem:optbds} and the fact that $G$ is $3$-colorable, the PTAS will output an $a$-rooted binary tree $T$ of size at least
\begin{align*}
\tp{1-\frac{\eps}{4}}(N-n^2). 
\end{align*}
We may assume that $T$ is a maximal binary tree in $D(G,\eps)$ (if not, then add more vertices to $T$ until we cannot add any further). Maximality ensures that the nodes $v_i$ are in the tree $T$ and moreover, the in-degree of $v_i$ in $T$ is exactly $2$. For each $v_i \in V$, let $c_i$ be the unique node among $\set{r_i\itn{1}, g_i\itn{1}, b_i\itn{1}}$ that is not in $T$. We define a coloring $\sigma: V \rightarrow \set{Red, Green, Blue}$ of $G$ as
\begin{align*}
\forall v_i \in V, \quad \sigma(v_i) = \begin{cases}
Red & \textup{if $c_i = r_i\itn{1}$} \\
Green & \textup{if $c_i = g_i\itn{1}$} \\
Blue & \textup{if $c_i = b_i\itn{1}$}.
\end{cases}
\end{align*}

We now argue that the coloring is proper for at least $(1-\eps)$-fraction of the edges of $G$. Suppose we have an edge $e=\set{v_i, v_j}$ which is monochromatic under $\sigma$, and suppose w.l.o.g. $\sigma(v_i) = \sigma(v_j) = Red$. This means that neither $r_i\itn{1}$ nor $r_j\itn{1}$ is included in $T$. Therefore $a_e^R \notin T$ since neither of the two vertices with incoming arcs from $a_e^R$ are in $T$. By Lemma \ref{lem:lb-bigtree}, we know that at most $\eps m$ vertices among $\cup_{e\in E}\{a_e^R,a_e^G,a_e^B\}$ can be excluded from $T$. Hence, the coloring $\sigma$ that we obtained can violate at most $\eps m$ edges in $G$.
\end{proof}

Finally, we prove Theorem \ref{theorem:dagmaxBT-no-const-approx} using the self-improving argument (Theorem \ref{thm:d-ptas}) and the APX-hardness of \dagmaxBT (Theorem \ref{theorem:dagmaxBT-no-ptas}).

\begin{proof}[Proof of Theorem \ref{theorem:dagmaxBT-no-const-approx}]
\begin{enumerate}
\item
We observe that the graph $G^2$ constructed in Section 
\ref{sec:hardness-dags}
for the self-improving reduction is a DAG if $G$ is a DAG. Therefore, by Theorem \ref{thm:d-ptas}, a polynomial-time constant-factor approximation for \dagmaxBT would imply a PTAS for \dagmaxBT, a contradiction to APX-hardness shown in Theorem \ref{theorem:dagmaxBT-no-ptas}. 

\item
Next we show hardness under the Exponential Time Hypothesis. Suppose there is a polynomial-time algorithm $\+A$ for \dagmaxBT that achieves an $\exp\tp{-C \cdot \log_2{n}/\log_2\log_2{n}}$-approximation for some constant $C > 0$. Given the input graph $G$ with $n-1$ vertices, let $k$ be an integer that satisfies
\begin{align*}
2^{\sqrt{n}} \le n^{2^k} \le 2^{2\sqrt{n}},  
\end{align*}
and run $\+A$ on $G^{2^k}$ to obtain a binary tree with size at least 
\begin{align*}
\exp\tp{-C \cdot \log_2{N}/\log_2\log_2{N}} OPT\tp{G^{2^k}},
\end{align*}
where $N = n^{2^k}$ upper bounds the size of $G^{2^k}$. Recursively running the algorithm suggested in Theorem \ref{thm:d-ptas} $k$ times gives us a binary tree in $G$ with size at least
\begin{align*}
   & \exp\tp{-C \cdot \frac{\log_2{N}}{\log_2\log_2{N} \cdot 2^k}} OPT(G) - 1 \\
\ge  & \exp\tp{-C \cdot \frac{2\sqrt{n}}{\log_2{\sqrt{n}}} \cdot \frac{\log_2{n}}{\sqrt{n}}} OPT(G) - 1 \\
\ge& \exp\tp{-4C} OPT(G) - 1 \ge \frac{1}{2} \cdot \exp\tp{-4C} OPT(G).
\end{align*}
The last inequality holds as long as 
\begin{align*}
OPT(G) \ge 2 \cdot e^{4C}.
\end{align*}
We note that if $OPT(G)$ is smaller than $2e^{4C}$ which is a constant, we can solve the problem exactly by brute force in polynomial time. Otherwise the above procedure can be regarded as a constant-factor approximation for \dagmaxBT. The running time is polynomial in 
\begin{align*}
N = n^{2^k} = \exp\tp{O\tp{\sqrt{n}}},
\end{align*}
which is sub-exponential. Moreover, from item \ref{theorem:dagmaxBT-no-const-approx_1} we know that it is $\NP$-hard to approximate \dagmaxBT within a constant factor, thus $\NP \subseteq \DTIME{\exp\tp{O\tp{\sqrt{n}}}}$.

\item
The proof of this item is almost identical to the previous one except that we choose a different integer $k$. Suppose there is an algorithm $\+A'$ for \dagmaxBT that achieves a $\exp\tp{-C \cdot \log^{1-\eps}{n}}$-approximation for some constant $C > 0$, and runs in time $\exp\tp{O\tp{\log^d{n}}}$ for some constant $d>0$. We show that there is an algorithm that achieves a constant-factor approximation for \dagmaxBT, and runs in time $\exp\tp{O\tp{\log^{d/\eps}{n}}}$.

Given a DAG $G$ on $n-1$ vertices as input for \dagmaxBT, let $k=\ceilfit{\tp{\frac{1}{\eps}-1}\log_2\log{n}}$ be an integer that satisfies
\begin{align*}
\tp{2^k \log{n}}^{1-\eps} \le 2^k \le 2\tp{\log{n}}^{\frac{1}{\eps} - 1}.  
\end{align*} 
Running $\+A'$ on $G^{2^k}$ gives us a binary tree with size at least 
\begin{align*}
\exp\tp{-C \cdot \log^{1-\eps}{N}} OPT\tp{G^{2^k}},
\end{align*}
where $N = n^{2^k}$ upper bounds the size of $G^{2^k}$. Recursively running the algorithm suggested in Theorem \ref{thm:d-ptas} $k$ times gives us a binary tree in $G$ with size at least
\begin{align*}
   & \exp\tp{-C \cdot \frac{\log^{1-\eps}{N}}{2^k}} OPT(G) - 1 \\
\ge  & \exp\tp{-C \cdot \frac{\tp{2^k\log{n}}^{1-\eps}}{2^k}} OPT(G) - 1 \\
\ge& \exp\tp{-C} OPT(G) - 1 \ge \frac{1}{2} \cdot \exp\tp{-C} OPT(G).
\end{align*}

The last inequality holds as long as 
\begin{align*}
OPT(G) \ge 2 \cdot e^C.
\end{align*}
We note that if $OPT(G)$ is smaller than $2e^C$ which is a constant, we can solve the problem exactly by brute force in polynomial time. Otherwise the above procedure can be regarded as a constant-factor approximation for \dagmaxBT. The running time is quasi-polynomial in $N$, i.e. for some constant $C' > 0$, the running time is upper-bounded by 
\begin{align*}
\exp\tp{C'\tp{\log^d{N}}} = \exp\tp{C'\tp{\tp{2^k\log{n}}^d}} \le \exp\tp{C'\tp{\log^{d/\eps}{n}}}.
\end{align*}

\end{enumerate}
\end{proof}


%% file: dag-lp-integrality-gap.tex
\subsection{An IP and its integrality gap for DAGs} \label{sec:IP-gap-dags}

Let $G=(V,E)$ with root $r\in V$ be the input graph. We use indicator variables $Y_u$ for the nodes $u \in V$ and $X_e$ for the arcs $e \in E$ to determine the set of nodes and arcs chosen in the solution. With these variables, the objective is to maximize the number of chosen nodes. Let $\delta^{out}(u)$ and $\delta^{in}(u)$ be the set of incoming and outgoing edges of $u$ respectively.
Constraints (\ref{eq:in-degree}) ensure that each chosen node has at most two incoming arcs. 
Constraints (\ref{eq:out-degree}) ensure that each chosen non-root node has an outgoing arc. 
Constraints (\ref{eq:cut-constraints}) are cut constraints that ensure that every subset $S$ of vertices containing a chosen node $u$ but not the root has at least one outgoing arc. 

\begin{align}
& \text{maximize}    & \sum_{v\in V} Y_v \notag              &\\
& \text{subject to}  & \sum_{e\in \delta^{in}(u)} X_e       &\le 2 Y_u &\quad\quad& \forall\ u \in V, \label{eq:in-degree} &&\\
& & \sum_{e\in \delta^{out}(u)} X_e        &= Y_u &\quad\quad &\forall\ u\in V \setminus \{r\},                                  && \label{eq:out-degree} &&\\
&                    & \sum_{e\in \delta^{out}(S)} X_e        &\ge Y_u && \forall\ u \in S\subset V\setminus \{r\}, \label{eq:cut-constraints} &&\\
&                    & 0 \le Y_u                             &\le 1                                 && \forall\ u \in V, \label{eq:vertex-bounds} &&\\
&                    & 0 \le X_e                             &\le 1                                 && \forall\ e \in E, \label{eq:edge-bounds} &&\\
&                    & Y\in \Z^{|V|}, X&\in \Z^{|E|}.                                && \label{eq:integrality} &&
\end{align}

We note that a similar IP formulation can also be written for the longest $s\rightarrow t$ path problem by replacing the factor $2$ in the RHS of (\ref{eq:in-degree}) with a factor of $1$. It can be shown that extreme point solutions for the LP-relaxation of such an IP are in fact integral. Owing to the similarity between the longest $s\rightarrow t$ path problem in DAGs and \dagmaxBT (as degree bounded maximum subtree problems), it might be tempting to conjecture that LP-based techniques might be helpful for \dagmaxBT. However, in contrast to the LP-relaxation for longest $s\rightarrow t$ path problem in DAGs (which is integral), the LP-relaxation of the above IP (by removing Constraints~\ref{eq:integrality}) for \dagmaxBT has very large integrality gap. 

\begin{theorem}
The integrality gap of the LP-relaxation of the above IP, even in DAGs, is $\Omega(n^{1/3})$, where $n$ is the number of nodes in the input DAG. 
\end{theorem}
\begin{proof}
We construct an integrality gap graph $T_k$ recursively as shown in Figure \ref{fig:gap-construction-recursive} with the base graph $T_1$ being a single node labeled $r_1$. We will denote the root vertices of $T_1, T_2, \ldots, T_{k-1}, T_k$ to be \emph{special} vertices. 
The layered construction and the direction of the arcs illustrate that the graph $T_n$ is a DAG. The number $V_k$ of nodes in the graph $T_k$ satisfies the recursion
\[
V_k=8V_{k-1}+13
\]
with $V_1=1$. Thus, $V_k=(13/7)(8^{k-1}-1)$. 

\begin{figure}[h]
	\centering
    \begin{subfigure}[b]{0.7\textwidth}
		\centering
        \includegraphics[scale=0.32]{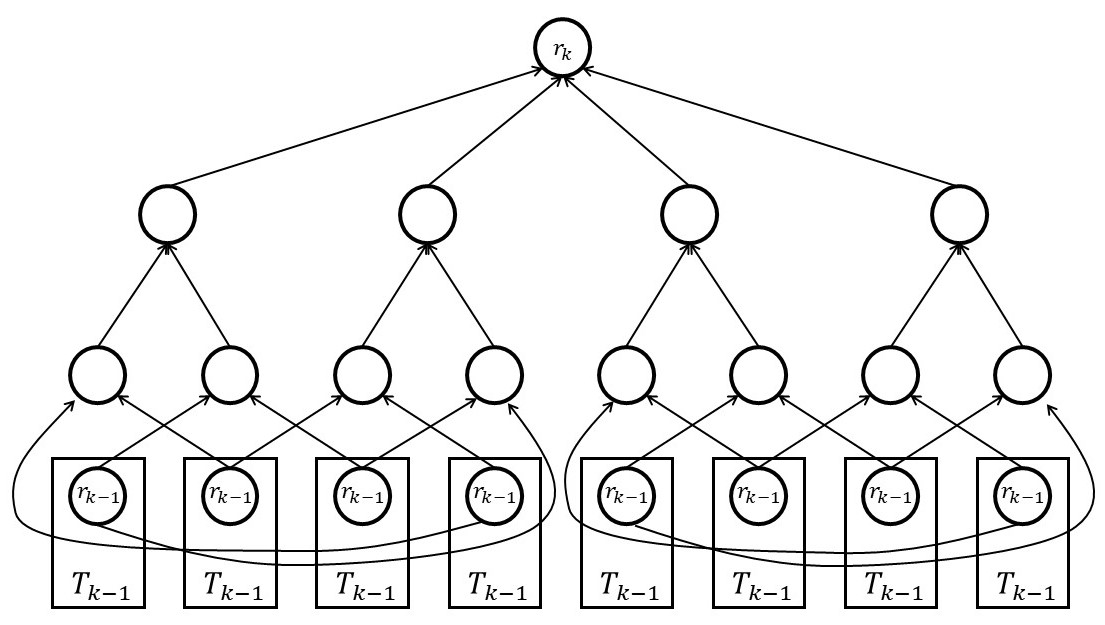}
        \caption{Integrality Gap Graph $T_k$}
    \label{fig:gap-construction-recursive}
    \end{subfigure}
	\quad
    \begin{subfigure}[b]{0.26\textwidth}  
        \centering
        \includegraphics[scale=0.27]{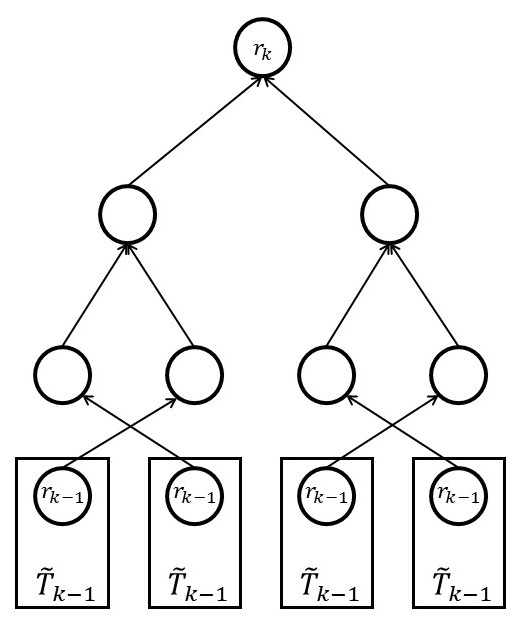}
        \caption{\centering Integral Optimum \newline Solution: $\tilde{T}_k$}
        \label{fig:gap-construction-IP-opt}
    \end{subfigure}
    \caption{DAG Integrality Gap}
    \label{fig:gap-construction}
\end{figure}

Due to the degree constraints, the optimal integral solution $\tilde{T}_k$ is obtained using a recursive construction as shown in Figure \ref{fig:gap-construction-IP-opt} with the base graph $\tilde{T}_1=T_1$. The number of arcs in the optimal integral solution satisfies the recursion
\[
\text{IP-OPT}(k)=4\text{IP-OPT}(k-1)+7
\]
with $\text{IP-OPT}(1)=0$. Thus, $\text{IP-OPT}(k)=(7/3)(4^{k-1}-1)$. 

In order to show large integrality gap, we give an LP-feasible solution with large objective value. The LP-feasible solution that we consider is $X_e:=1/2$ for every arc $e$ in the graph with 
\begin{equation*}
Y_{u}:=
\begin{cases}
1 \text{ if $u$ is a special vertex,}\\
\frac{1}{2} \text{ otherwise.}
\end{cases}
\end{equation*}
We now argue that this solution satisfies all constraints of the LP-relaxation. The in-degree and out-degree constraints hold by definition. We now show that the cut constraints, i.e., constraints (\ref{eq:cut-constraints}), are satisfied. We have two cases: 

\noindent \emph{Case 1.} Suppose $u$ is a non-special vertex. For every non-special vertex $u$, we have a path from $u$ to the root $r$. So, every cut $S$ containing $u$ but not $r$ has an arc leaving it and hence $\sum_{e\in \delta^{out}(S)}X_e\ge 1/2 = Y_u$. 

\noindent \emph{Case 2.} Suppose $u$ is a special verteex. For every special vertex $u$, we have two arc-disjoint paths from $u$ to the root $r$. So, every cut $S$ containing $u$ but not $r$ has at least $2$ arcs leaving it and hence $\sum_{e\in \delta^{out}(S)}X_e\ge 1 = Y_u$. 

The objective value of this LP-feasible solution satisfies the recursion
\[
\text{LP-obj}(k)=8\text{LP-obj}(k-1)+14
\]
with $\text{LP-obj}(1)=0$. Thus, $\text{LP-obj}(k)=2(8^{k-1}-1)$. Consequently, the integrality gap of the LP for instance $T_k$ is $\Omega(2^{k-1})=\Omega(V_k^{1/3})$. 
\end{proof}

%% file: undirected-hardness.tex
\section{Hardness results for undirected graphs}\label{sec:undirected-hardness}

We show the inapproximability of finding a maximum binary tree in undirected graphs. We use $OPT(G)$ to denote the size (number of vertices) of a maximum binary tree in $G$.

\subsection{Self-improvability}
This section is devoted to proving the following theorem.

\begin{theorem} \label{thm:ud-ptas}
If \umaxBT has a polynomial-time algorithm that achieves a constant-factor approximation, then it has a PTAS.
\end{theorem} 

The idea is to gradually boost up the approximation ratio by running the constant-factor approximation algorithm on squared graphs. This idea was also followed by Karger, Motwani and Ramkumar \cite{karger1997approximating} for the longest path problem. 
For our purpose, we need a different construction of squared graphs. We introduce our construction now. 

\begin{definition} \label{def:ud-squared}
For an undirected graph $G$, its squared graph $G^{\XBox 2}$ is defined as the graph obtained by performing the following operations in $G$.
\begin{enumerate}
\item Replace each edge $\set{u,v} \in E(G)$ with a copy $G_{u,v}$ of $G$. Connect $u$ and $v$ to all vertices in $G_{u,v}$. We will refer to these copies as \emph{edge} copies.
\item For each vertex $v \in V(G)$, introduce two copies of $G$ denoted by $G_v\itn{1}$ and $G_v\itn{2}$, and connect $v$ to all vertices in $G_v\itn{1}$ and $G_v\itn{2}$. We will refer to these copies as \emph{pendant} copies.
\end{enumerate}
We will use $V(G)$ to denote the original vertices of $G$ and the same vertices in the graph $G^{\XBox 2}$ (see Figure \ref{figure:undir-squared-graph} for an example). 
We define $G^{\XBox 2^{k+1}}$ recursively as $G^{\XBox 2^{k+1}} \coloneqq \tp{G^{\XBox 2^k}}^{\XBox 2}$ with the base case $G^{\XBox 1} \coloneqq G$.
\end{definition}

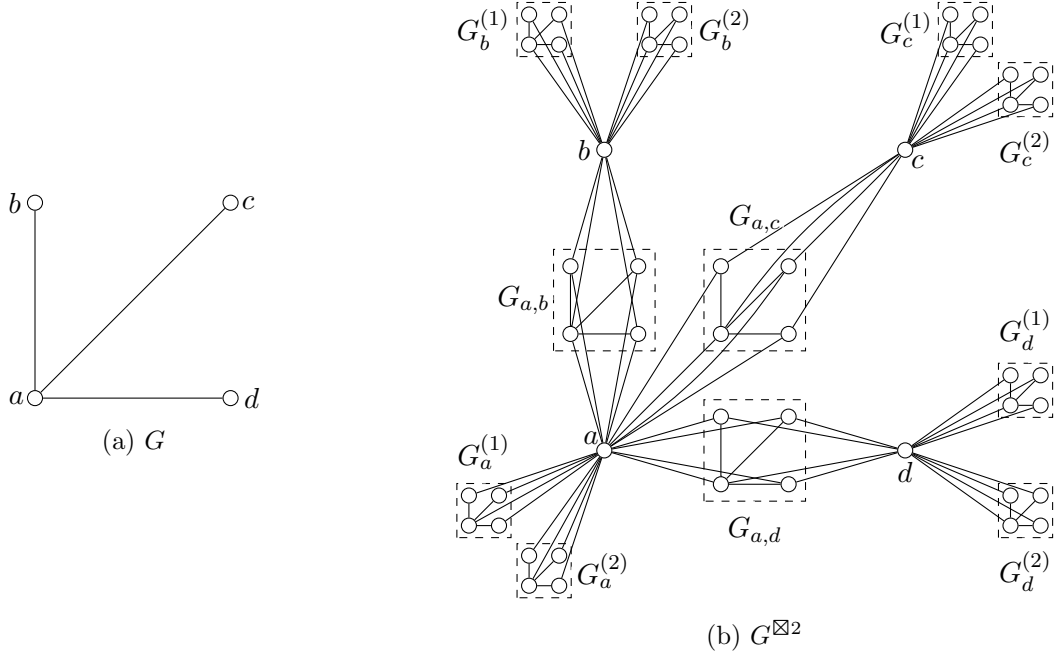
\begin{figure}[h]
\centering
\begin{subfigure}{.5\textwidth}
\centering
\begin{tikzpicture}[every node/.style={draw=black, circle, inner sep=0pt, minimum size=0.2cm}]
\def \r {1.3cm}
\node[label=left:$a$] (v0) at (0,0) {};
\node[label=left:$b$] (v1) at (0, 2*\r) {};
\node[label=right:$c$] (v2) at ({2*\r}, {2*\r}) {};
\node[label=right:$d$] (v3) at (2*\r, 0) {};
\draw (v0) to (v1) (v0) to (v2) (v0) to (v3);
\end{tikzpicture}
\caption{$G$}
\end{subfigure}%
\begin{subfigure}{.5\textwidth}
\centering
\begin{tikzpicture}[every node/.style={draw=black, circle, inner sep=0pt, minimum size=0.2cm}]
\def \r {2cm}
\def \l {0.45cm}
\def \p {1.5*\l}
\node[label=135:$a$] (v0) at (0,0) {};
\node[label=left:$b$] (v1) at (0, 2*\r) {};
\node[label=315:$c$] (v2) at ({2*\r}, {2*\r}) {};
\node[label=below:$d$] (v3) at (2*\r, 0) {};
\foreach \x in {0,1}
    \foreach \y in {0,1}
        \node (1\x\y) at ({(1-2*\x)*\l}, {\r+(1-2*\y)*\l}) {};
\foreach \x in {0,1}
    \foreach \y in {0,1}
        \node (2\x\y) at ({\r+(1-2*\x)*\l}, {\r+(1-2*\y)*\l}) {};
\foreach \x in {0,1}
    \foreach \y in {0,1}
        \node (3\x\y) at ({\r+(1-2*\x)*\l}, {(1-2*\y)*\l}) {};
\foreach \i in {1,3}
    \foreach \x in {0,1}
        \foreach \y in {0,1}
            \draw (\i\x\y) to (v\i) (\i\x\y) to (v0) (\i\x\y) to (\i 11) ;
\draw (201) to (v2) (210) to (v2) (200) to (v2);
\draw (211) to[in=215, out=55] (v2);
\draw (201) to (v0) (210) to (v0) (211) to (v0);
\draw (200) to[in=35, out=235] (v0);
\draw (211) to (201) (211) to (200) (211) to (210);

\draw[dashed] (-\p, \r-\p) to node[draw=white, midway, left] {$G_{a,b}$} (-\p, \r+\p) to (\p, \r+\p) to (\p, \r-\p) to cycle;
\draw[dashed] (\r-\p, \r-\p) to (\r-\p, \r+\p) to node[draw=white, midway, above] {$G_{a,c}$} (\r+\p, \r+\p) to (\r+\p, \r-\p) to cycle;
\draw[dashed] (\r-\p, -\p) to (\r-\p, \p) to (\r+\p, \p) to (\r+\p, -\p) to node[draw=white, midway, below] {$G_{a,d}$} cycle;

\def \rr {0.8cm}
\def \l {0.2cm}
\def \p {1.8*\l}
\foreach \x in {0,1} {
    \foreach \y in {0,1} {
        \node (a1\x\y) at ({-2*\rr+(1-2*\x)*\l}, {-\rr+(1-2*\y)*\l}) {};
        \node (a2\x\y) at ({-\rr+(1-2*\x)*\l}, {-2*\rr+(1-2*\y)*\l}) {};
    }
}
\foreach \x in {0,1} {
    \foreach \y in {0,1} {
        \draw (a1\x\y) to (v0);
        \draw (a2\x\y) to (v0);
    }
}
\foreach \i in {1,2} {
    \draw (a\i11) to (a\i01);
    \draw (a\i11) to (a\i10);
    \draw (a\i11) to (a\i00);
}
\draw[dashed] ({-2*\rr-\p}, {-\rr-\p}) to ({-2*\rr-\p}, {-\rr+\p}) to node[draw=white, midway, above] {$G_a\itn{1}$} ({-2*\rr+\p}, {-\rr+\p}) to ({-2*\rr+\p}, {-\rr-\p}) to cycle;
\draw[dashed] ({-\rr-\p}, {-2*\rr-\p}) to ({-\rr-\p}, {-2*\rr+\p}) to ({-\rr+\p}, {-2*\rr+\p}) to node[draw=white, midway, right] {$G_a\itn{2}$} ({-\rr+\p}, {-2*\rr-\p}) to cycle;

\foreach \x in {0,1} {
    \foreach \y in {0,1} {
        \node (b1\x\y) at ({-\rr+(1-2*\x)*\l}, {2*\r+2*\rr+(1-2*\y)*\l}) {};
        \node (b2\x\y) at ({ \rr+(1-2*\x)*\l}, {2*\r+2*\rr+(1-2*\y)*\l}) {};
    }
}
\foreach \x in {0,1} {
    \foreach \y in {0,1} {
        \draw (b1\x\y) to (v1);
        \draw (b2\x\y) to (v1);
    }
}
\foreach \i in {1,2} {
    \draw (b\i11) to (b\i01);
    \draw (b\i11) to (b\i10);
    \draw (b\i11) to (b\i00);
}
\draw[dashed] ({-\rr-\p}, {2*\r+2*\rr-\p}) to node[draw=white, midway, left]{$G_b\itn{1}$} ({-\rr-\p}, {2*\r+2*\rr+\p}) to ({-\rr+\p}, {2*\r+2*\rr+\p}) to ({-\rr+\p}, {2*\r+2*\rr-\p}) to cycle;
\draw[dashed] ({\rr-\p}, {2*\r+ 2*\rr-\p}) to ({\rr-\p}, {2*\r+2*\rr+\p}) to ({\rr+\p}, {2*\r+2*\rr+\p}) to node[draw=white, midway, right]{$G_b\itn{2}$} ({\rr+\p}, {2*\r+2*\rr-\p}) to cycle;

\foreach \x in {0,1} {
    \foreach \y in {0,1} {
        \node (c1\x\y) at ({2*\r+2*\rr+(1-2*\x)*\l}, {2*\r+  \rr+(1-2*\y)*\l}) {};
        \node (c2\x\y) at ({2*\r+  \rr+(1-2*\x)*\l}, {2*\r+2*\rr+(1-2*\y)*\l}) {};
    }
}
\foreach \x in {0,1} {
    \foreach \y in {0,1} {
        \draw (c1\x\y) to (v2);
        \draw (c2\x\y) to (v2);
    }
}
\foreach \i in {1,2} {
    \draw (c\i11) to (c\i01);
    \draw (c\i11) to (c\i10);
    \draw (c\i11) to (c\i00);
}
\draw[dashed] ({2*\r+  \rr-\p}, {2*\r+2*\rr-\p}) to node[draw=white, midway, left]{$G_c\itn{1}$} ({2*\r+  \rr-\p}, {2*\r+2*\rr+\p}) to ({2*\r+  \rr+\p}, {2*\r+2*\rr+\p}) to ({2*\r+  \rr+\p}, {2*\r+2*\rr-\p}) to cycle;
\draw[dashed] ({2*\r+2*\rr-\p}, {2*\r+  \rr-\p}) to ({2*\r+2*\rr-\p}, {2*\r+  \rr+\p}) to ({2*\r+2*\rr+\p}, {2*\r+  \rr+\p}) to ({2*\r+2*\rr+\p}, {2*\r+  \rr-\p}) to node[draw=white, midway, below]{$G_c\itn{2}$} cycle;

\foreach \x in {0,1} {
    \foreach \y in {0,1} {
        \node (d1\x\y) at ({2*\r+2*\rr+(1-2*\x)*\l}, {-\rr+(1-2*\y)*\l}) {};
        \node (d2\x\y) at ({2*\r+2*\rr+(1-2*\x)*\l}, { \rr+(1-2*\y)*\l}) {};
    }
}
\foreach \x in {0,1} {
    \foreach \y in {0,1} {
        \draw (d1\x\y) to (v3);
        \draw (d2\x\y) to (v3);
    }
}
\foreach \i in {1,2} {
    \draw (d\i11) to (d\i01);
    \draw (d\i11) to (d\i10);
    \draw (d\i11) to (d\i00);
}
\draw[dashed] ({2*\r+2*\rr-\p}, {\rr-\p}) to ({2*\r+2*\rr-\p}, {\rr+\p}) to node[draw=white, midway, above]{$G_d\itn{1}$} ({2*\r+2*\rr+\p}, {\rr+\p}) to ({2*\r+2*\rr+\p}, {\rr-\p}) to cycle;
\draw[dashed] ({2*\r+2*\rr-\p}, {-\rr-\p}) to ({2*\r+2*\rr-\p}, {-\rr+\p}) to ({2*\r+2*\rr+\p}, {-\rr+\p}) to ({2*\r+2*\rr+\p}, {-\rr-\p}) to node[draw=white, midway, below]{$G_d\itn{2}$} cycle;

\end{tikzpicture}
\caption{$G^{\XBox 2}$}
\end{subfigure}
\caption{A graph $G$ and its squared graph $G^{\XBox 2}$}
\label{figure:undir-squared-graph}
\end{figure}

Given an undirected graph $G$ with $n$ vertices, the number of vertices in $G^{\XBox 2^k}$ is upper bounded by $n^{3^k}$ since at most $|E(G)|+2|V(G)| \le n^2$ copies of $G$ are introduced in $G^{\XBox 2}$ when $n \ge 3$.

For a binary tree $T$ and $d \in \set{0,1,2,3}$, let $I_d(T) \subseteq V(T)$ be the set of nodes with degree exactly $d$. The following lemma is our main tool in the reduction.

\begin{lemma} \label{lemma:external-nodes}
For every non-empty binary tree $T$, we have 
\[
3|I_0(T)| + 2|I_1(T)| + |I_2(T)| = |V(T)| + 2,
\]
\end{lemma}
\begin{proof}
We prove by induction on $|V(T)|$. When $|V(T)| = 1$, the tree $T$ has one degree-0 node. Therefore $|I_0(T)| = 1$ and $|I_1(T)| = |I_2(T)|=0$, thus $3|I_0(T)|+2|I_0(T)|+|I_1(T)| = 3 = |V(T)| + 2$ holds.

Suppose that the statement holds for all binary trees with $t-1$ nodes for $t \ge 2$. Let $T$ be a binary tree with $t$ nodes. When $t \ge 2$ there are no degree-0 nodes, so we only need to verify that $2|I_1(T)| + |I_2(T)| = |V(T)| + 2$. Let $\ell$ be an arbitrary leaf node in $T$ and $p$ be its unique neighbor. Then $\deg_T(\ell)=1$. Removing $\ell$ results in a binary tree $T'$ with $t-1$ nodes. We have the following cases: 
\begin{enumerate}
\item If $p \in I_3(T)$, then $I_1(T') = I_1(T) \setminus \set{\ell}$ and $I_2(T') = I_2(T) \cup \set{p}$. 
\item If $p \in I_2(T)$, then $I_1(T') = \tp{I_1(T) \setminus \set{\ell}}\cup \set{p}$ and $I_2(T') = I_2(T) \setminus \set{p}$. 
\item If $p \in I_1(T)$, then $I_0(T')=\set{p}$, $I_1(T')=I_1(T) \setminus \set{p, \ell}$ and $I_2(T') = I_2(T)$.
\end{enumerate}
In all cases, we have
\begin{align*}
2|I_1(T)| + |I_2(T)| = 3|I_0(T')| + 2|I_1(T')| + |I_2(T')| + 1 = |V(T')| + 3 = t + 2
\end{align*}
where the second equality is due to the induction hypothesis. This completes the inductive step.
\end{proof}

The next lemma shows that $OPT(G)$ is super-multiplicative under the squaring operation. 

\begin{lemma} \label{lem:ud-squaredopt}
$OPT\tp{G^{\XBox 2}} \ge 2OPT\tp{G}^2 + 2OPT\tp{G}$.
\end{lemma}
\begin{proof}
Let $T_1$ be an optimal binary tree in $G$, i.e. $|V(T_1)| = OPT(G)$. We construct a binary tree $T_2$ in $G^{\XBox 2}$ as follows:
\begin{enumerate}
\item For $\set{u, v} \in E\tp{G}$, let $T_{u,v}$ be the optimal binary tree (identical to $T_1$) in the edge copy $G_{u,v}$. For $v \in V\tp{G}$ and $i \in \set{1,2}$, let $T_v\itn{i}$ be the optimal binary tree in the pendant copy $G_v\itn{i}$.
\item For every edge $\set{u,v} \in T_1$, add $T_{u,v}$ into $T_2$ along with two edges $\set{u, \ell}$ and $\set{v, \ell}$ where $\ell$ is an arbitrary leaf node in $T_{u,v}$. This step generates $|V(T_1)| + |V(T_1)| \cdot (|V(T_1)|-1) = |V(T_1)|^2$ nodes in $T_2$, since the number of edges in $E(T_1)$ is $|V(T_1)|-1$. 
\item For $v \in I_1(T_1)$, add both $T_v\itn{1}$ and $T_v\itn{2}$ to $T_2$ by connecting $v$ to $\ell\itn{1}$ and $\ell\itn{2}$, where $\ell\itn{1}$ and $\ell\itn{2}$ are arbitrary leaf nodes in $T_v\itn{1}$ and $T_v\itn{2}$, respectively. For $u \in I_2(T_1)$, add only $T_u\itn{1}$ to $T_2$ by connecting $u$ to a leaf node in $T_u\itn{1}$. By Lemma \ref{lemma:external-nodes}, this step generates $|V(T_1)| \cdot (|V(T_1)|+2)$ nodes in $T_2$.
\end{enumerate}
Since $T_1$ is a binary tree, it follows that $T_2$ is a binary tree. Moreover, the size of $V(T_2)$ is 
\begin{align*}
|V(T_2)| = |V(T_1)|^2 + |V(T_1)| \cdot (|V(T_1)|+2) = 2OPT\tp{G}^2 + 2OPT\tp{G},
\end{align*} 
which cannot exceed $OPT\tp{G^{\XBox 2}}$.
\end{proof}

The next two lemmas show that a large binary tree in $G^{\XBox 2}$ can be used to obtain a large binary tree in $G$.

\begin{lemma} \label{lem:BSTrecover}
Given $T_2$ as a binary tree in $G^{\XBox 2}$, there is a linear-time (in the size of $G^{\XBox 2}$) algorithm that finds a binary tree $T'_1$ in $G$ with vertex set $V(T_2) \cap V(G)$.
\end{lemma}
\begin{proof}
Given a binary tree $T_2$ in the squared graph $G^{\XBox 2}$, the algorithm finds a binary tree $T'_1$ in $G$ by going through every edge $\set{u, v} \in E(G)$ and adding it to $T'_1$ whenever the unique path from $u$ to $v$ in $T_2$ goes through the edge copy $G_{u,v}$. We discard the edge $\set{u, v}$ if there does not exist a path through $G_{u,v}$ connecting $u$ and $v$ in $T_2$. 

By construction, the subgraph returned by the algorithm has maximum degree 3 and is acyclic. Moreover, it is connected since the path between any two nodes $u, v \in T_2$ is preserved in $T'_1$. Therefore, $T'_1$ is a binary tree in $G$ with vertex set $V(T_2) \cap V(G)$.
\end{proof}

\begin{lemma} \label{lem:ud-sqrtfactor}
For every $\alpha \in (0,1]$, given a binary tree $T_2$ in $G^{\XBox 2}$ with size
\begin{align*}
|V\tp{T_2}| \ge \alpha OPT\tp{G^{\XBox 2}} - 1,
\end{align*}
there is a linear-time (in the size of $G^{\XBox 2}$) algorithm that finds a binary tree $T_1$ in $G$ with size 
\begin{align*}
|V\tp{T_1}| \ge \sqrt{\alpha} OPT\tp{G} - 1.
\end{align*}
\end{lemma}
\begin{proof}
Running the algorithm suggested in Lemma \ref{lem:BSTrecover} gives us a binary tree $T_1'$ in $G$ with vertex set $V\tp{T_1'} = V\tp{T_2} \cap V(G)$ in linear time. Therefore if $|V\tp{T_2} \cap V(G)| \ge \sqrt{\alpha} OPT(G) > \sqrt{\alpha} OPT(G) - 1$ then the lemma is already proved. So, we may assume that $|V\tp{T_2} \cap V(G)| < \sqrt{\alpha} OPT(G)$.

Let us first deal with the case when $V\tp{T_2} \cap V\tp{G} = \varnothing$. In this case $T_2$ completely resides within some edge copy or pendant copy of $G^{\XBox 2}$. That means $T_2$ is already a binary tree in $G$ with size  
\begin{align*}
\abs{V\tp{T_2}} \ge \alpha OPT\tp{G^{\XBox 2}} - 1 \ge \tp{\sqrt{\alpha} OPT\tp{G}}^2 - 1 \ge \sqrt{\alpha} OPT\tp{G} - 1,
\end{align*}
where the second inequality uses Lemma \ref{lem:ud-squaredopt} and the third inequality assumes $\sqrt{\alpha}OPT\tp{G} \ge 1$.

Suppose $T_1'$ is not the empty tree. Removing vertices in $V(G)$ from $T_2$ results in a forest with each connected component of the forest residing completely within one of the copies of $G$ (either edge or pendent copy). Let $\+F$ be the set of trees in this forest. 
\begin{claim}\label{claim:few-trees}
We have 
\[
|\+F|< 2\sqrt{\alpha} OPT(G) + 1.
\]
\end{claim}
\begin{proof}
Each tree $T\itn{j} \in \+F$ is connected to one or two vertices in $V\tp{T_2} \cap V(G)$. Let $\+F_j$ be the set of trees connected to exactly $j$ vertices in $V(G) \cap V\tp{T_2}$. Then $\+F = \+F_2 \cup \+F_1$. We now bound the size of $\+F_2$ and $\+F_1$. If $T\itn{j}$ is connected to two vertices $u, v \in V(G) \cap V\tp{T_2}$, then $\set{u, v}$ must be an edge in $T_1'$. Moreover, if there are two distinct trees $T\itn{j_1}$, $T\itn{j_2}$ both connected to $u, v \in V\tp{G} \cap V\tp{T_2}$ then there will be a cycle. Therefore $|\+F_2| \le |E\tp{T_1'}| = |V\tp{T_1'}|-1$. As to $\+F_1$, for $j \in \set{0,1,2,3}$, every vertex $v \in I_j(T_1')$ is connected to at most $3-j$ trees in $\+F_1$. Every vertex $v \in V(G)\setminus  V\tp{T_1'}$ is not connected to any tree in $\+F_1$. Therefore $|\+F_1| \le 3|I_0(T_1')|+2|I_1(T_1')| + |I_2(T_1')|$. This gives an upper bound on $|\+F|$.
\begin{align*}
|\+F| &= |\+F_1| + |\+F_2| \le |V\tp{T_1'}|-1 + 3|I_0(T_1')| + 2|I_1(T_1')| + |I_2(T_1')| \\
 &= 2|V\tp{T_1'}| + 1 = 2|V\tp{T_2} \cap V(G)| + 1 < 2\sqrt{\alpha} OPT(G) + 1,
\end{align*}
where the second equality is due to Lemma \ref{lemma:external-nodes}. 
\end{proof}
\begin{claim}
There exists a tree $T^* \in \+F$ with size at least
\begin{align*}
|V\tp{T^*}| \ge \sqrt{\alpha} OPT(G) - 1.
\end{align*}
\end{claim}
\begin{proof}
Suppose not. Then every tree $T\in \+F$ has $|V\tp{T}|<\sqrt{\alpha} OPT(G) - 1$. Then, 
\begin{align*}
|V\tp{T_2}| &= |V\tp{T_2} \cap V(G)| + \sum_{T\itn{j} \in \+F}\abs{V\tp{T\itn{j}}} \\
 &< \sqrt{\alpha} OPT(G) +\tp{\sqrt{\alpha} OPT(G) - 1} \tp{2\sqrt{\alpha} OPT(G) + 1} \quad \quad \quad \quad \text{(By Claim \ref{claim:few-trees})} \\
 &= 2\alpha OPT(G)^2 - 1 \\
 &< \alpha \tp{2OPT(G)^2 + 2OPT(G)} - 1 \\
 &\le \alpha OPT\tp{G^{\XBox 2}} - 1
\end{align*}
which is a contradiction. The last inequality here is due to Lemma \ref{lem:ud-squaredopt}.
\end{proof}

Since $T^*$ resides in one of the copies of $G$, it is a binary tree in $G$. A DFS would find $T^*$ in linear time.
\end{proof}

Now we are ready prove the main theorem of this subsection.

\begin{proof}[Proof of Theorem \ref{thm:ud-ptas}]
Suppose that we have a polynomial-time algorithm $\+A$ that achieves an $\alpha$-approximation for \umaxBT. Given an undirected graph $G$ and $\eps > 0$, let
\begin{align*}
k := 1 + \ceilfit{ \log_2\frac{\log_2 \alpha}{\log_2 (1-\eps) }}
\end{align*}
be an integer constant that depends on $\alpha$ and $\eps$. We construct $G^{\XBox 2^k}$ and run algorithm $\+A$ on it. We get a binary tree in $G^{\XBox 2^k}$ of size at least $\alpha OPT\tp{G^{\XBox 2^k}} - 1$. By Lemma \ref{lem:ud-sqrtfactor}, we can obtain a binary tree in $G$ of size at least  
\begin{align*}
\alpha^{2^{-k}} OPT(G) - 1 \ge \alpha^{2^{-k+1}} OPT(G) \ge (1-\eps) OPT(G). 
\end{align*}
The first inequality holds as long as
\begin{align*}
OPT(G) \ge \frac{1}{\sqrt{1-\eps} - (1-\eps)} \ge \frac{1}{\alpha^{2^{-k}} - \alpha^{2^{-k+1}}}.
\end{align*}
We note that if $OPT(G)$ is smaller than $1/\tp{\sqrt{1-\eps} - (1-\eps)}$ which is a constant, then we can solve the problem exactly by brute force in polynomial time. Finally, we also observe that for fixed constant $\eps$, the running time of this algorithm is polynomial since there are at most $n^{3^k}=n^{O(1)}$ vertices in the graph $G^{\XBox 2^k}$.
\end{proof}

\subsection{APX-hardness}

In this section, we show that \umaxBT is APX-hard. The reduction is from the following problem, denoted as \textsc{TSP(1,2)}.

\begin{problem}{TSP(1,2)}
	Given: A complete undirected graph $K_n$ with edge weights $w_{ij}$ where $w_{ij} \in \set{1,2}$ $\forall i,j \in [n]$.
	
	Goal: A tour with minimum weight which starts and finishes at the same vertex and visits every other vertex exactly once.
\end{problem}

For an instance $K_n=([n], E_1 \cup E_2)$ of \textsc{TSP(1,2)} where $E_1$ is the set of edges with weight 1, and $E_2$ is the set of edges with weight 2, define two subgraphs $S_1=([n],E_1)$ and $S_2=([n],E_2)$. It is convenient to think of $S_1$ and $S_2$ as unweighted graphs.

\begin{theorem}[\cite{papadimitriou1993traveling}, Theorem 9 of \cite{karger1997approximating}]
\textsc{TSP(1,2)} is APX-hard even on instances with optimal value $n$, i.e. instances whose associated subgraph $S_1$ has a Hamiltonian cycle.
\end{theorem}

The following lemma shows that a binary tree in $K_n$ with a small number of degree-3 vertices can be transformed into a path without too much increase in total weight.

\begin{lemma} \label{lem:TreeToPath}
Let $K_n$ be a weighted complete graph with edge weights in $\set{1, 2}$. Let $T$ be a binary tree in $K_n$ with $\ell$ nodes and let $w$ be the sum of the edge weights of $T$. If at most $d$ vertices in $T$ have degree 3, then in linear time we can find a path in $K_n$ with length $\ell$ and total weight at most $w+d$.
\end{lemma}
\begin{proof}
Pick an arbitrary root node $r$ with $\deg_T(r) \le 2$ and consider the $r$-rooted tree $T$. Perform the following operation on $T$ bottom-up to convert every vertex of degree-$3$ in $T$ into a vertex of degree-$2$. 
\begin{enumerate}
\item Find a vertex $v \in T$ with degree 3 such that the two subtrees of $v$ are paths. Suppose that the path on the left is $(a_1, a_2, \cdots, a_p)$ and the path on the right is $(b_1, b_2, \cdots, b_q)$ with $a_1$ and $b_1$ being adjacent to $v$ in $T$.
\item Merge the two paths into a longer path $(a_1, \cdots, a_p, b_1, b_2, \cdots, b_q)$ by replacing edge $(v, b_1)$ with $(a_p, b_1)$.
\end{enumerate}
This operation converts $v$ into a degree-$2$ vertex and increases the total weight of $T$ by at most 1 (in case of $w(v, b_1)=1$ and $w(a_p, b_1)=2$). Since no new degree-$3$ vertices are introduced during the operation, the total number of degree-$3$ vertices decreases by one. The path claimed by the lemma is thus obtained by recursively performing this operation $d$ times. We note that the final path is effectively a post-order traversal of $T$.
\end{proof}

We also need the following structural result. 
\begin{lemma}\label{lemma:lot-of-pendants-few-degree-3}
Let $H=(V,E)$ be a graph with $n$ vertices and let $\tilde{H}=(V\cup V', E\cup E')$ where $V':=\{v':v\in V\}$ and $E':=\{\{v,v'\}: v\in V\}$. Suppose we have a binary tree $\tilde{T}$ in $\tilde{H}$ of size at least $(2-\eps)n$. Then, the graph $T$ obtained by restricting $\tilde{T}$ to $H$ is a binary tree with at most $\eps n$  vertices of degree $3$. 
\end{lemma}
\begin{proof}
Let us denote the set $V'$ of added nodes as pendants. Since the pendants have degree $1$ in $\tilde{H}$, the restricted graph $T$ is a binary tree. We now show that the number of vertices in $T$ with degree $3$ is small. We note that the number of nodes in $\tilde{H}$ is $2n$ and hence $\tilde{T}$ contains all but at most $\eps n$ vertices of $\tilde{H}$. For every vertex $v$ with $\deg_T(v)=3$, its pendant $v'$ is not in $\tilde{T}$. Thus, in order for $\tilde{T}$ to contain all but $\eps n$ vertices of $\tilde{H}$, the number of vertices of degree $3$ in $T$ cannot exceed $\eps n$. 
\end{proof}
	
The following lemma implies that if there is a PTAS for \umaxBT then there is a PTAS for \textsc{TSP(1,2)}.

\begin{lemma} \label{lem:ud-apxhard}
Suppose there is a PTAS for \umaxBT (even restricted to graphs that contain binary spanning trees), then for every $\eps > 0$ there is a polynomial-time algorithm which takes as input an undirected complete graph $K_n$ with edge weights in $\set{1,2}$ whose associated subgraph $S_1$ has a Hamiltonian cycle to output a tour with weight at most $(1+\eps)n$.
\end{lemma}
\begin{proof}
Let $K_n$ be the input instance of \textsc{TSP(1,2)} with $S_1$ and $S_2$ defined as above. Let $\tilde{S}_1$ be the graph constructed from $S_1$ as follows: For every $i \in [n]$, introduce a new vertex $v_i$ adjacent to vertex $i$ in $S_1$. We will refer to $v_i$ as the pendant of vertex $i$.

We note that $\tilde{S}_1$ has a spanning binary tree (i.e. $OPT(\tilde{S}_1) = 2n$) because $S_1$ has a Hamiltonian cycle. Therefore we can run the PTAS for \umaxBT on graph $\tilde{S}_1$ with error parameter $\eps/4$ (which still takes polynomial time). The PTAS would output a binary tree $\tilde{T}$ of size at least 
\begin{align*} 
\tp{1-\frac{\eps}{4}} \cdot OPT(\tilde{S}_1) = \tp{1-\frac{\eps}{4}} \cdot 2n = 2n - \frac{\eps}{2}n,
\end{align*}
By Lemma \ref{lemma:lot-of-pendants-few-degree-3}, we obtain a binary tree $T$ in $S_1$ with the following properties: 
\begin{enumerate}
\item $T$ contains at least $2n - \eps n/2 - n = \tp{1-\eps/2}n$ vertices and
\item $T$ has at most $\eps n/2$ vertices with degree 3.
\end{enumerate}

By Lemma \ref{lem:TreeToPath} we can transform $T$ into a path in $K_n$ that has total weight at most $(1-\eps/2)n + \eps n/2 = n$ and contains at least $(1-\eps/2)n$ vertices. This path can be extended to a valid tour by including the missing vertices using edges with weight 2. Such a tour will have total weight at most $n + 2(\eps n/2) = (1+\eps)n$.
\end{proof}

Now we are ready to prove Theorem \ref{theorem:umaxBT-no-const-approx}.

\umaxBTnoApprox*
\begin{proof}
\mbox{}
\begin{enumerate}
\item
Suppose there is a polynomial-time algorithm for \umaxBT that achieves a constant-factor approximation. By Theorem \ref{thm:ud-ptas}, the problem also has a PTAS. By Lemma \ref{lem:ud-apxhard}, \textsc{TSP(1,2)} would also have a PTAS, thus contradicting its APX-hardness. Therefore \umaxBT does not admit a polynomial-time constant-factor approximation assuming $\P \neq \NP$.

\item
Next we show hardness under the exponential time hypothesis. Suppose there is a polynomial-time algorithm $\+A$ for \umaxBT that achieves an $\exp\tp{-C\cdot (\log{n})^{\log_3{2}-\eps}}$-approximation for constants $C, \eps>0$. We show that there is an algorithm that achieves a constant-factor approximation for \umaxBT, and runs in time $\exp\tp{O\tp{(\log{n})^{\eps^{-1}\log_3{2}}}}$.

Given an undirected graph $G$ on $n$ vertices as input for \umaxBT, let $k = \ceilfit{\tp{{\eps^{-1}\log_3{2}-1}}\log_3\log{n}}$, which is an integer that satisfies
\begin{align*}
\tp{3^k \log n}^{\log_3{2}-\eps} \le 2^k \text{ and } 3^{k} \le 3\tp{\log{n}}^{\eps^{-1}\log_3{2}-1}.
\end{align*}
Running $\+A$ on $G^{\XBox 2^k}$ gives us a binary tree with size at least
\begin{align*}
\exp\tp{-C\cdot (\log{N})^{\log_3{2}-\eps}} OPT\tp{G^{\XBox 2^k}}
\end{align*}
where $N \le n^{3^k}$ is the size of $G^{\XBox 2^k}$. Recursively applying Lemma \ref{lem:ud-sqrtfactor} gives us a binary tree in $G$ with number of vertices being at least
\begin{align*}
   & \exp\tp{-C\cdot \frac{(\log{N})^{\log_3{2}-\eps}}{2^k}} OPT(G) - 1 \\
\ge  & \exp\tp{-C\cdot \frac{(3^k\log{n})^{\log_3{2}-\eps}}{2^k}} OPT(G) - 1 \\
\ge& \exp\tp{-C} OPT(G) - 1 \ge \frac{1}{2} \cdot \exp\tp{-C} OPT(G).
\end{align*}
The last inequality holds as long as
\begin{align*}
OPT(G) \ge 2 \cdot e^C.
\end{align*}
We note that when $OPT(G)$ is smaller than $2e^C$ which is a constant, we can solve the problem exactly by brute force in polynomial time. Otherwise the above procedure can be regarded as a constant-factor approximation for \umaxBT. The running time is quasi-polynomial in $N$, i.e. for some constant $C' > 0$, the running time is upper-bounded by 
\begin{align*}
\exp\tp{C'\tp{\log^d{N}}} \le \exp\tp{C'\tp{\tp{3^k \cdot \log{n}}^d}} \le \exp\tp{C'\tp{(\log{n})^{\eps^{-1}d\log_3{2}}}}.
\end{align*}
Moreover, from item \ref{theorem:umaxBT-no-const-approx_1} we know that it is $\NP$-hard to approximate \umaxBT within a constant factor, therefore $\NP \subseteq \DTIME{\exp\tp{O\tp{(\log{n})^{\eps^{-1}d\log_3{2}}}}}$.
\end{enumerate}
\end{proof}

We remark that APX-hardness for \umaxBT on graphs with spanning binary trees does not rule out constant-factor approximation algorithms on such instances. This is because our squaring operation might lose spanning binary trees ($G^{\XBox 2}$ does not necessarily contain a spanning binary tree even if $G$ does).

%% file: bipartite-permutation-graphs.tex
\section{An efficient algorithm for bipartite permutation graphs} \label{sec:bipartite-permutation}

In this section we prove Theorem \ref{theorem:bipartite-permutation-graphs}. We begin with some structural properties of bipartite permutation graphs that will be helpful in designing the algorithm.

\subsection{Structural properties of bipartite permutation graphs}

\begin{definition}
A \emph{strong ordering} $\sigma$ of a bipartite graph $G=(S,T,E)$ is an ordering of $S$ and an ordering of $T$ such that
\begin{align*}
\forall s <_\sigma s' \in S, t <_\sigma t' \in T, \quad (s, t') \in E \textup{ and } (s', t) \in E \implies (s, t) \in E \textup{ and } (s', t') \in E. 
\end{align*}
\end{definition}
Informally, strong ordering essentially states that the existence of \emph{cross edges} implies the existence of \emph{parallel edges}. 
The following theorem from \cite{spinrad1987bipartite} shows that strong ordering exactly characterizes bipartite permutation graphs.

\begin{theorem}[Theorem 1 of \cite{spinrad1987bipartite}] \label{thm:strong-ordering}
A bipartite graph $G=(S,T,E)$ is also a permutation graph if and only if $G$ has a strong ordering.
\end{theorem}

\begin{corollary} \label{cor:biperm-indsubgraph}
Let $G$ be a bipartite permutation graph and $H=G[V_H]$ be an induced subgraph. Then $H$ is also a bipartite permutation graph.
\end{corollary}
\begin{proof}
Let $\sigma$ be a strong ordering of $G$. The corollary follows by applying Theorem \ref{thm:strong-ordering} and observing that the projection of $\sigma$ onto $V_H$ is a strong ordering of $H$.
\end{proof}

In the following, when we are given a bipartite permutation graph $G=(S,T,E)$ along with a strong ordering $\sigma$ (or simply a strongly ordered bipartite permutation graph), we always assume that the elements in $S$ and $T$ are sorted in ascending order according to $\sigma$:
\begin{align*}
s_1 <_\sigma s_2 <_\sigma \cdots <_\sigma s_{p}, \quad t_1 <_\sigma t_2 <_\sigma \cdots <_\sigma t_{q}.
\end{align*}
Here $p \coloneqq |S|$ and $q \coloneqq |T|$.

The following lemma shows that in a bipartite permutation graph the neighborhood of a vertex $v \in G$ has a nice consecutive structure.

\begin{lemma} \label{lem:neighborhood-interval}
Let $G=(S,T,E)$ be a connected bipartite permutation graph and $\sigma$ be a strong ordering of $G$. For every $s_i \in S$, there exist $a_i \le b_i \in [q]$ such that 
\begin{align*}
N(s_i) = [t_{a_i}, t_{b_i}] \coloneqq \set{t_{a_i}, t_{a_i+1}, \cdots, t_{b_i-1}, t_{b_i}}.
\end{align*}
Moreover, for any $s_i, s_j \in S$ such that $s_i <_\sigma s_j$ we have
\begin{align*}
a_i \le a_j, \quad b_i \le b_j.
\end{align*}
\end{lemma}
\begin{proof}
For the first part, let $s \in S$ be an arbitrary vertex and let $t_1$, $t_2$ be the smallest and largest elements in $N(s)$ (with respect to $\sigma$), respectively. Consider any $t \in T$ satisfying $t_1 <_\sigma t <_\sigma t_2$. We want to show that $t \in N(s)$. Since $G$ is connected, there must be some $s' \in S$ adjacent to $t$. Suppose $s' <_\sigma s$. Since $\sigma$ is a strong ordering, we have
\begin{align*}
(s', t_2) \in E \textup{ and } (s, t) \in E \implies (s, t) \in E.
\end{align*} 
A symmetric argument holds for the case $s <_\sigma s'$. Therefore $t \in N(s)$.

For the second part, suppose for the sake of contradiction that $a_j < a_i$. Recall that $s_i <_\sigma s_j$, we have
\begin{align*}
(s_i, t_{a_i}) \in E \textup{ and } (s_j, t_{a_j}) \in E \implies (s_i, t_{a_j}) \in E.
\end{align*}
That means $t_{a_j} \in N(s_i) = [t_{a_i}, t_{b_i}]$, contradicting with $a_j < a_i$. A symmetric argument can be used to prove $b_i \le b_j$.
\end{proof}

Another important property of (connected) bipartite permutation graphs is that they contain crossing-free spanning trees. 

\begin{definition}
Given a bipartite permutation graph $G=(S,T,E)$ and a strong ordering $\sigma$ of $G$, we say a subgraph $H$ has an \emph{edge crossing} (w.r.t. the strong ordering $\sigma$) if $H$ contains two edges $(s_1, t_1)$ and $(s_2, t_2)$ such that $s_1 <_\sigma s_2$ and $t_2 <_\sigma t_1$. Otherwise we say $H$ is \emph{crossing-free}.
\end{definition}

We need the following theorem from \cite{smith2011minimum}.

\begin{theorem}[Corollary 4.19 of \cite{smith2011minimum}] \label{thm:crossing-free}
Let $G$ be a strongly ordered connected bipartite permutation graph. There exists a minimum degree spanning tree (MDST) of $G$ which is crossing-free.
\end{theorem}

\begin{lemma} \label{lem:mbt-crossing-free}
Let $G$ be a strongly ordered bipartite permutation graph. There exists a maximum binary tree in $G$ which is crossing-free.
\end{lemma}
\begin{proof}
Consider any maximum binary tree $H=(V_H,E_H)$ in $G$ and the induced subgraph $G[V_H]$. By Corollary \ref{cor:biperm-indsubgraph} we have that $G[V_H]$ is a bipartite permutation graph. Moreover, $G[V_H]$ contains a spanning binary tree and is thus connected. By Theorem \ref{thm:crossing-free} there is a crossing-free MDST of $G[V_H]$, which we will denote by $H'$. We note that $H'$ is a binary tree, and that $H'$ has the same size as $H$ since they are both spanning trees of $G[V_H]$. Therefore $H'$ is a maximum binary tree in $G$ which is crossing-free. 
\end{proof}

The next lemma inspires the definition of subproblems which lead us to the Dynamic Programming based algorithm.

\begin{lemma} \label{lem:first-edge}
Let $G = (S, T, E)$ be a strongly ordered bipartite permutation graph, and let $H=\tp{V_H, E_H}$ be a connected crossing-free subgraph of $G$. Let $s_1$ and $t_1$ be the two minimum vertices (w.r.t. the strong ordering) in $S \cap V_H$ and $T \cap V_H$, respectively. Then we have $\set{s_1, t_1} \in E_H$, and that one of $s_1$ and $t_1$ has degree 1. 
\end{lemma}
\begin{proof}
Suppose for the sake of contradiction that $\set{s_1, t_1} \notin E_H$. Since $H$ is connected, there exists a path $\tp{s_1, t_2, \cdots, s_2, t_1}$ where $t_2 >_\sigma t_1$ and $s_2 >_\sigma s_1$. However, the two edges $\set{s_1, t_2}$ and $\set{s_2, t_1}$ constitute an edge crossing which is a contradiction to the assumption that $H$ is crossing-free.

We proved that $\set{s_1, t_1} \in E_H$. Suppose both $s_1$ and $t_1$ have at least one more neighbors, say $t_2$ and $s_2$, respectively, then once more $\set{s_1, t_2}$ and $\set{s_2, t_1}$ constitute an edge crossing. Therefore one of $s_1$ and $t_1$ has degree 1.
\end{proof}

\subsection{The algorithm}

In this section, we give a dynamic programming approach for solving \umaxBT on bipartite permutation graphs. We first focus on connected, strongly ordered bipartite permutation graphs. Theorem \ref{theorem:bipartite-permutation-graphs} will follow from the fact that a strong ordering can be found in linear time.

Let $G=(S,T,E)$ be a strongly ordered bipartite permutation graph with $|S|=p$ and $|T|=q$. For $i \in [p], j \in [q]$, define
\begin{align*}
S_i \coloneqq \set{s_i, s_{i+1}, \cdots, s_p}, \quad T_j \coloneqq \set{t_j, t_{j+1}, \cdots, t_q}.
\end{align*}
We also use the convention $S_{p+1} = T_{q+1} = \varnothing$. Define $[i,j] \coloneqq G[S_i \cup T_j]$, i.e. the subgraph of $G$ induced by $S_i \cup T_j$.

For $i \in [p]$ and $j \in [q]$, let $\mbt_S(i,j)$ be the maximum cardinality (number of edges) of a binary tree $H$ in $[i,j]$ such that 
\begin{enumerate}
\item $H$ is crossing-free, 
\item $\set{s_i, t_j} \in E_H$, 
\item $s_i$ is a leaf node in $H$. 
\end{enumerate}
$\mbt_T(i,j)$ is similarly defined except that in the last constraint we require $t_j$ (instead of $s_i$) to be a leaf node in $H$. 
Finally let 
\begin{align*}
\mbt(G) \coloneqq \max_{i \in [p], j \in [q]}\set{\max\set{\mbt_S(i,j), \mbt_T(i,j)}}.
\end{align*}

\begin{lemma} \label{lem:dp-correctness}
Let $G=(S,T,E)$ be a strongly ordered bipartite permutation graph. Let $\opt(G)$ be the cardinality of the maximum binary tree in $G$. Then $\mbt(G) = \opt(G)$.
\end{lemma}
\begin{proof}
Since it is trivial that $\opt(G) \ge \mbt(G)$, we focus on the other direction $\opt(G) \le \mbt(G)$.

Let $H=(V_H,E_H)$ be a maximum binary tree in $G$, i.e. $|E_H| = \opt(G)$. By Lemma \ref{lem:mbt-crossing-free}, we can further assume that $H$ is a crossing-free. Let $s_i$ be the minimum vertex in $S \cap V_H$ and let $t_j$ be the minimum vertex in $T \cap V_H$. Since $H$ is a connected crossing-free subgraph, by Lemma \ref{lem:first-edge} we have that $\set{s_i,t_j} \in E_H$, and that one of $s_i$ and $t_j$ is a leaf node in $H$. Observing that $H$ is also a maximum binary tree in the subgraph $[i,j]$, we have
\begin{align*}
\opt(G) = |E_H| = \max\set{\mbt_S(i,j), \mbt_T(i,j)} \le \mbt(G).
\end{align*}
\end{proof}

Now in order to compute $\opt(G)$, it suffices to compute $\mbt(G)$ which amounts to solving the subproblems $\mbt_S(i,j)$ and $\mbt_T(i,j)$. The following recurrence relations immediately give a dynamic programming algorithm:
\begin{align}
\mbt_S(i,j) = \begin{cases}
0 & \textup{if $s_i \notin N_{[i,j]}(t_j)$,} \\
1 & \textup{if $N_{[i,j]}(t_j) = \set{s_i}$,} \\
{\displaystyle \max\set{\mbt_T(i+1,j)+1, \max_{2\le k\le \abs{N_{[i,j]}(t_j)}-1}\set{\mbt_T(i+k,j)+2}} } & \textup{if $\abs{N_{[i,j]}(t_j)} \ge 2$.}
\end{cases} \label{eqn:recurrence1} \\
\mbt_T(i,j) = \begin{cases}
0 & \textup{if $t_j \notin N_{[i,j]}(s_i)$,} \\
1 & \textup{if $N_{[i,j]}(s_i) = \set{t_j}$,} \\
{\displaystyle \max\set{\mbt_S(i,j+1)+1, \max_{2\le k\le \abs{N_{[i,j]}(s_i)}-1}\set{\mbt_S(i,j+k)+2}} } & \textup{if $\abs{N_{[i,j]}(s_i)} \ge 2$.} 
\end{cases} \label{eqn:recurrence2}
\end{align}
The boundary conditions are given by
\begin{align*}
\mbt_S(p+1,q+1) = \mbt_T(p+1,q+1) = 0.
\end{align*}

\begin{lemma} \label{lem:recurrence-correctness}
$\mbt_S(i,j)$ and $\mbt_T(i,j)$ satisfy the recurrence relations (\ref{eqn:recurrence1}) and (\ref{eqn:recurrence2}).
\end{lemma}
\begin{proof}
Since $S$ and $T$ are symmetric, we will only prove that $\mbt_S(i,j)$ satisfies relation (\ref{eqn:recurrence1}).

\emph{Case 1}: $s_i \notin N_{[i,j]}(t_j)$. Clearly $\mbt_S(i,j) = 0$ since $\set{s_i, t_j} \notin E$ implies that constraint 2 cannot be satisfied by any binary tree.

\emph{Case 2:} $N_{[i,j]}(t_j) = \set{s_i}$, i.e. $s_i$ is the unique neighbor of $t_j$ in the graph $[i,j]$. Since by constraint 3 vertex $s_i$ has to be a leaf node in $H$, we know that $\set{s_i, t_j}$ is the only binary tree which satisfies all 3 constraints. In this case $\mbt_S(i,j) = 1$.

\emph{Case 3:} Case 1 and Case 2 do not occur, which implies $d \coloneqq \abs{N_{[i,j]}(t_j)} \ge 2$. Consider the optimal binary tree $H=(V_H, E_H)$ satisfying all 3 constraints. Let $s_{i+k}$ ($1 \le k \le d-1$) be the ``furthest'' neighbor of $t_j$, i.e. the maximal element in $N_H(t_j) \setminus \set{s_i}$. We further consider 2 possibilities.
\begin{itemize}
\item $k=1$. We note that $t_j$ is a degree-2 node in this case. Consider the binary tree $H' = (V_{H'}, E_{H'})$ where $V_{H'}=V_H \setminus \set{s_i}$ and $E_{H'}=E_H \setminus \set{\set{s_i, t_j}}$. $H'$ is a feasible solution to the subproblem $\mbt_T(i+1,j)$ since $H'$ is a crossing-free binary tree which contains $t_j$ as a leaf node and the edge $\set{s_{i+1},t_j}$. We deduce that $|E_H| = |E_{H'}| + 1 \le \mbt_T(i+1,j) + 1$.

\item $k \ge 2$. Since $H$ is maximum, $t_j$ must have another neighbor other than $s_i$ and $s_{i+k}$. By Lemma \ref{lem:neighborhood-interval}, $s_{i+\ell}$ is a neighbor of $t_j$ for any $0 \le \ell \le k$. Since $H$ is crossing-free, that third neighbor of $t_j$ is a leaf node in $H$. Therefore without loss of generality we can assume that it is $s_{i+1}$. Consider the binary tree $H' = (V_{H'}, E_{H'})$ where $V_{H'}=V_H \setminus \set{s_i, s_{i+1}}$ and $E_{H'}=E_H \setminus \set{\set{s_i, t_j}, \set{s_{i+1}, t_j}}$. $H'$ is a feasible solution to the subproblem $\mbt_T(i+k,j)$ since $H'$ is a crossing-free binary tree which contains $t_j$ as a leaf node and the edge $\set{s_{i+k},t_j}$. We deduce that $|E_H| = |E_{H'}| + 2 \le \mbt_T(i+k,j) + 2$. 
\end{itemize}
Thus, we conclude that 
\begin{align*}
\mbt_S(i,j) \le \max\set{\mbt_T(i+1,j)+1, \max_{2\le k\le d-1}\set{\mbt_T(i+k,j)+2}}.
\end{align*}
To see the other direction of the inequality, we note that a feasible solution to $\mbt_T(i+1,j)$ induces a feasible solution to $\mbt_S(i,j)$ by including the edge $\set{s_i, t_j}$, and a feasible solution to $\mbt_T(i+k, j)$ induces a feasible solution to $\mbt_S(i,j)$ by including the edges $\set{s_i, t_j}$ and $\set{s_{i+1},t_j}$.
\end{proof}

We now give a formal proof of Theorem \ref{theorem:bipartite-permutation-graphs}.

\bipermALG*
\begin{proof}
Given a bipartite permutation graph $G$ with $n$ vertices and $m$ edges, there is an $O(m+n)$ time algorithm for finding a strong ordering of $G$ (cf. \cite{spinrad1987bipartite}). Suppose $G$ has connected components $G_1, G_2, \cdots, G_\ell$ and $G_i$ has $n_i$ vertices, hence $\sum_{i=1}^{\ell}n_i = n$. Every $G_i$ is a (strongly ordered) connected bipartite permutation graph. Since any binary tree in $G$ completely resides in one connected component of $G$, it suffices to solve $\mbt\tp{G_i}$ for every $G_i$ and return $\max_{1\le i \le\ell}\set{\mbt\tp{G_i}}$. Solving $\mbt\tp{G_i}$ requires $O\tp{n_i^3}$ time since there are $O\tp{n_i^2}$ subproblems ($\mbt_S\tp{i, j}$ and $\mbt_T\tp{i, j}$ for $i, j \in [n_i]$) solving each of which requires $O\tp{n_i}$ time. The overall running time is $O\tp{\sum_{i=1}^{\ell}n_i^3} = O\tp{n^3}$.
\end{proof}

%% file: conclusions.tex
\section{Conclusion and Open Problems}\label{sec:conclusions}




In this work, we introduced the maximum binary tree problem  (MBT) and presented hardness of approximation results for undirected, directed, and directed acyclic graphs, a fixed-parameter algorithm with the solution as the parameter, and efficient algorithms for 
bipartite permutation graphs. 
Our work raises several open questions that we state below. \\

\noindent {\bf Inapproximability of \dmaxBT.} 
The view that MBT is a variant of the longest path problem leads to the natural question of whether the inapproximability results for MBT match that of longest path:  
Is MBT in directed graphs (or even in DAGs) hard to approximate within a factor of $1/n^{1-\varepsilon}$ (we recall that longest path is hard to approximate within a factor of $1/n^{1-\eps}$ \cite{BHK04})? We remark that the self-improving technique is weak to handle $1/n^{1-\eps}$-approximations since the squaring operation yields no improvement.
The reduction in~\cite{BHK04} showing $1/n^{1-\eps}$-inapproximability of longest paths is from a restricted version of the vertex-disjoint paths problem and is very specific to paths. 
Furthermore, directed cycles play a crucial role in their reduction for a fundamental reason: longest path is polynomial-time solvable in DAGs. 
However, it is unclear if directed cycles are the source of hardness for MBT in digraphs (since MBT is already hard in DAGs). \\

\noindent {\bf Bicriteria Approximations.}
Given our inapproximability results, one natural algorithmic possibility is that of bicriteria approximations: can we find a tree with at least $\alpha \cdot OPT$ vertices while violating the degree bound by a factor of at most $\beta$? In particular, this motivates an intriguing direction concerning the longest path problem: Given an undirected/directed graph $G$ with a path of length $k$, can we find a $c_1$-degree tree in $G$ with at least $k/{c_2}$ vertices for some constants $c_1$ and $c_2$ efficiently?\\

\noindent {\bf Maximum Binary Tree in Permutation DAGs.}
Finally, it would be interesting to resolve the complexity of MBT in permutation DAGs (and permutation graphs). This would also resolve the open problem posed by Byers, Heeringa, Mitzenmacher, and Zervas of whether the maximum heapable subsequence problem is solvable in polynomial time \cite{byers2010heapable}.

%% file: acknowledgement.tex
\section*{Acknowledgements}

Karthekeyan Chandrasekaran is supported by NSF CCF-1814613 and NSF CCF-1907937. 
Elena Grigorescu, Young-San Lin, and Minshen Zhu are supported by NSF CCF-1910659 and NSF CCF-1910411. Gabriel Istrate was supported by a grant of Ministry of Research and Innovation, CNCS - UEFISCDI project number PN-III-P4-ID-PCE-2016-0842, within PNCDI III.